\documentclass[journal]{IEEEtran}
\usepackage{lineno}
\usepackage{cite}
\usepackage{hyperref}
\usepackage{amsmath,amssymb,amsfonts}
\usepackage{amsmath}
\usepackage{amsthm}
\DeclareMathOperator*{\argmax}{argmax}

\usepackage{array}
\usepackage{graphicx}
\usepackage{textcomp}
\usepackage{xcolor}
\usepackage{graphicx}
\usepackage{float}
\usepackage{subfigure}
\usepackage{amsmath}
\usepackage{amsfonts,amssymb}
\usepackage{mathrsfs}
\usepackage{mathtools}
\usepackage{multirow}
\usepackage{algpseudocode}
\makeatletter
\newif\if@restonecol
\makeatother

\usepackage[linesnumbered,ruled,vlined]{algorithm2e}
\usepackage{algpseudocode}

\usepackage{setspace}
\usepackage{footmisc}
\usepackage[justification=centering]{caption}
\DeclareMathOperator*{\argmin}{argmin}
\usepackage{mathtools}
\usepackage{dsfont}
\usepackage{bbm}
\newtheorem{remark}{Remark}
\newtheorem{corollary}{Corollary}

\theoremstyle{definition}
\newtheorem{theorem}{Theorem}

\newtheorem{lemma}{Lemma}

\makeatletter
\newcommand{\biggg}{\bBigg@{3}}
\newcommand{\Biggg}{\bBigg@{3.5}}
\makeatother
\usepackage{bm}
\makeatother
\begin{document}

\title{Spectral Efficiency Maximization for DMA-enabled Multiuser MISO with Statistical CSI\\}

\author{Hao Xu, Boyu Ning, Chongjun Ouyang, and Hongwen Yang
\vspace{-5pt}
\thanks{Hao Xu and Hongwen Yang are with the School of Information and Communication Engineering, Beijing University of Posts and Telecommunications, Beijing 100876, China (e-mail: Xu\_Hao@bupt.edu.cn; yanghong@bupt.edu.cn).}
\thanks{
B. Ning is with the National Key Laboratory of Wireless Communications, University of Electronic Science and Technology of China, Chengdu 611731, China (e-mail: boydning@outlook.com).}
\thanks{
Chongjun Ouyang was with the School of Electrical and Electronic Engineering, University College Dublin, Dublin, D04 V1W8, Ireland, and is now with the School of Electronic Engineering and Computer Science, Queen Mary University of London, London, E1 4NS, U.K. (e-mail: c.ouyang@qmul.ac.uk).}
}

\maketitle

\begin{abstract}
Dynamic metasurface antennas (DMAs) offer the potential to achieve large-scale antenna arrays with low power consumption and reduced hardware costs, making them a promising technology for future communication systems. This paper investigates the spectral efficiency (SE) of DMA-enabled multiuser multiple-input single-output (MISO) systems in both uplink and downlink transmissions, using only statistical channel state information (CSI) to maximize the ergodic sum rate of multiple users. For the uplink system, we consider two decoding rules: minimum mean square error (MMSE) with and without successive interference cancellation (SIC). For both decoders, we derive closed-form surrogates to substitute the original expressions of ergodic sum rate and formulate tractable optimization problems for designing DMA weights. Then, a weighted MMSE (WMMSE)-based algorithm is proposed to maximize the ergodic sum rate. For the downlink system, we derive an approximate expression for the ergodic sum rate and formulate a hybrid analog/digital beamforming optimization problem that jointly optimizes the digital precoder and DMA weights. A penalty dual decomposition (PDD)-based algorithm is proposed by leveraging the fractional programming framework. Numerical results validate the accuracy of the derived surrogates and highlight the superiority of the proposed algorithms over baseline schemes. It is shown that these algorithms are effective across various DMA settings and are particularly well-suited for system design in fast time-varying channels.
\end{abstract}

\vspace{-5pt}
\section{Introduction}
In recent years, the rapid development of fifth generation (5G) and the upcoming sixth generation (6G) mobile communication technologies has revolutionized global communication networks \cite{Andrews2014,Saad2020,Tataria2021}. As a key technology of the 5G network, massive multiple-input multiple-output (MIMO) has been proven to significantly enhance spectral efficiency (SE) and the throughput of wireless systems \cite{Larsson2014}. Nevertheless, high energy consumption, substantial hardware costs, and deployment size limitations pose challenges for implementing massive MIMO technology \cite{Rial2016,Mo2017}, particularly for the ultra-massive MIMO (UM-MIMO) envisioned for future 6G networks \cite{Ning2023}. In this context, researchers have been exploring innovative solutions to deploy massive MIMO systems without compromising performance or incurring prohibitive costs. Among these, dynamic metasurface antennas (DMA), enabled by advances in metamaterials, have emerged as a promising solution for achieving ultra-large-scale antenna arrays in future communication systems \cite{Shlezinger2019,Shlezinger2021}.

A DMA is a planar array composed of radiating metamaterial elements. These elements are typically implemented using structures such as microstrip lines, parallel plates, and cavities, which can interact with guided waves to facilitate signal transmission and reception. By controlling the tunable components integrated into each metamaterial element, one can achieve active beamforming of the transmitter or receiver. The initial implementation of DMAs was proposed in \cite{Sleasman2015} based on microstrip and then was extended to 2D architectures by composing waveguides and cavities \cite{Yoo2018}. In a 2D waveguide configuration, scattered waves from each element propagate in all directions, resulting in increased signal interference between ports. In the field of wireless communications, metamaterial components are not only utilized in DMA architectures but are also widely employed in reconfigurable intelligent surfaces (RIS), including passive RIS \cite{Wu2020,Huang2019,Han2019,Ning2021,Ouyang2022}, active RIS \cite{Zhang2023}, simultaneous transmitting and reflecting reconfigurable intelligent surfaces (STAR-RIS) \cite{Xu2021}, and multi-functional RIS \cite{Sun2024}, which can intelligently control wireless transmission environments between transmitters and receivers. To clearly distinguish between the concepts of DMA and RIS, it is important to highlight their key differences. Firstly, RIS is predominantly passive, meaning it does not connect to RF links and lacks baseband signal processing capabilities. Additionally, RIS is typically deployed at a location that is some distance away from both the base station (BS) and the users. This deployment necessitates at least two distinct channels: the user-to-RIS channel and the BS-to-RIS channel. In contrast, DMA is an active component that integrates with the RF systems at the transceivers. This integration allows DMA to facilitate both transmission and dynamic beamforming. Consequently, while DMA and RIS may share similar hardware structures, they serve distinctly different functions within communication systems.

{\color{black}The DMA-based transceivers can achieve hybrid analog/digital (A/D) beamforming more efficiently compared to traditional phase-shifter-based hybrid A/D architectures \cite{Shlezinger2021}. On the one hand, DMA eliminates the need for additional analog combining circuits, such as phase shifters, resulting in lower hardware costs. On the other hand, the tuning of DMA elements relies on simple components like varactor tubes, which yields reduced power consumption.}

Thanks to the ability of DMAs to enable large-scale antenna arrays with low cost and power consumption, their potential in massive MIMO systems has been extensively explored in recent years. For uplink transmission, the authors in \cite{Shlezinger2019} studied a DMA-assisted multi-user (MU) MIMO system and proposed two alternating optimization algorithms to maximize the achievable sum rate for both frequency-flat and frequency-selective channels. This work was extended to DMAs in MIMO orthogonal frequency division modulation (OFDM) receivers \cite{Wang2021}. For downlink MU-MIMO systems, authors in \cite{Wang2019} investigated the DMA tuning strategies, relaxing the Lorentzian constraint on DMA elements to make optimization tractable, which may reduce performance. To address this, authors in \cite{Kimaryo2023} developed an efficient algorithm without such relaxations. This study was further expanded to near-field communication scenarios \cite{Zhang2021,Zhang2022,Xu2024}. Specifically, authors in \cite{Zhang2021,Zhang2022} proposed a mathematical model to characterize near-field wireless channels in DMA-based downlink systems and investigated the sum rate maximization problem across three array architectures, i.e., fully digital, hybrid analog/digital, and DMA-based. Authors in \cite{Xu2024} explored the impact of spatial-wideband effect, near-field effect, and frequency selectivity on sum rate in wideband transmission. {\color{black}Besides, energy efficiency (EE), as a key performance metric \cite{Ge2022}, has been extensively studied in DMA assisted systems \cite{You2023,Chen2025}. Specifically, the authors in \cite{You2023} optimized EE performance in uplink multiuser systems using both instantaneous and statistical channel state information (CSI). This work was later extended to downlink multiuser systems with instantaneous CSI in \cite{Chen2025}.}

It is worth noting that most of the studies mentioned above are based on the assumption of instantaneous CSI to simplify the design. {\color{black}However, the large-scale dimension of antenna arrays leads to considerable computational complexity for real-time channel estimation. Additionally, the need for frequent reconfiguration of DMA elements results in increased power consumption and hardware degradation, especially in high mobility scenarios with fast time-varying channels.} Therefore, designing DMA-based systems that leverage statistical CSI is significant for practical systems. Up to now, research on DMA-assisted systems with statistical CSI is still limited. The authors in \cite{You2023} focused primarily on EE. A hybrid RIS and DMA-assisted MU multiple-input single-output (MISO) system has been studied based on full statistical CSI \cite{Zhang2024}, subjected to a constraint relaxation for DMA processing. Furthermore, the authors in \cite{Papazafeiropoulos2024} investigated a Stacked intelligent metasurface (SIM) aided MIMO system to maximize the downlink achievable rate with statistical CSI. Although SIM also implements holographic MIMO, it operates under different constraints than DMA. In summary, despite its importance, how to improve the SE performance in DMA-assisted systems with statistical CSI still remains an open problem.

To fill the research gap, in this paper, we investigate the SE performance of the DMA-enabled MU-MISO system for both uplink and downlink transmissions. As a primary work of this study, we consider a basic DMA architecture where the radiating elements are integrated with 1D waveguides coupled to RF chains on stacked microstrips\cite{Smith2017}. The main contributions of the paper are summarized as follows:
\begin{itemize}
\item For uplink systems, we aim to maximize the spectral efficiency (SE) performance under two classical decoding rules: minimum mean square error with successive interference cancellation (MMSE-SIC) decoding and without SIC (MMSE-nSIC) decoding. We derive a closed-form ergodic sum rate upper bound for the MMSE-SIC scenario and an approximation for the MMSE-nSIC scenario. Based on this, we formulate two tractable optimization problems concerning the phase shifts of the DMA and propose a weighted MMSE (WMMSE)-based algorithm. For optimizing the DMA weights, we introduce an element-wise refinement (EWR)-based method, which provides a closed-form solution for each step.
\item For downlink systems, we initially derived an approximate expression for the ergodic sum rate. Building on this, we formulate a hybrid analog/digital (A/D) beamforming optimization problem that incorporates coupling power constraints on the DMA weights and the digital precoder. To address this problem, we propose a penalty dual decomposition (PDD)-based algorithm by leveraging the fractional programming (FP) framework.

\item Numerical results validate the tightness and accuracy of all derived surrogates in our setup for a realistic communication environment. Compared to the existing benchmarks, the proposed algorithms have superiority in both uplink and downlink scenarios. For the uplink case, the proposed algorithm, relying only on statistical CSI, achieves over 90\% of the performance obtained with instantaneous CSI. For the downlink case, the proposed algorithm with statistical CSI can achieve 80\% of the performance with instantaneous CSI at lower SNR levels.
\end{itemize}

{\color{black}
The rest of this paper is structured as follows. Section \uppercase\expandafter{\romannumeral2} reviews the considered DMA architecture and presents the uplink and downlink system models based on statistical CSI. Section \uppercase\expandafter{\romannumeral3} analyzes and optimizes the uplink ergodic sum rate with statistical CSI, considering both MMSE-SIC and MMSE-nSIC combiners. Section \uppercase\expandafter{\romannumeral4} focuses on the analysis and optimization of the downlink ergodic sum rate. Section \uppercase\expandafter{\romannumeral5} provides numerical results to validate the analytical results and the effectiveness of the proposed optimization algorithms. Finally, Section \uppercase\expandafter{\romannumeral6} concludes the paper.  A comparison with related works is given in TABLE \ref{table1}, and the notations utilized throughout the paper are explained in TABLE \ref{table2}.}

\begin{table*}[!t]
\caption{Comparison with existing work.}
\label{table1}
\centering
{\color{black}\begin{tabular}{|>{\raggedright\arraybackslash}p{3cm} |>{\raggedright\arraybackslash}p{2cm}|>{\raggedright\arraybackslash}p{1cm} |>{\raggedright\arraybackslash}p{1cm}|>{\raggedright\arraybackslash}p{1cm}
|>{\raggedright\arraybackslash}p{1cm}|>{\raggedright\arraybackslash}p{1cm} |>{\raggedright\arraybackslash}p{1cm}|>{\raggedright\arraybackslash}p{1cm} |>{\raggedright\arraybackslash}p{1cm}|}
\hline
& \textbf{Our paper} & \cite{Shlezinger2019} & \cite{Wang2019} & \cite{Kimaryo2023} & \cite{Zhang2021} & \cite{Zhang2022} & \cite{Xu2024} & \cite{You2023} & \cite{Chen2025} \\ \hline
Uplink & $\checkmark$ & $\checkmark$ & &  &   &   & $\checkmark$ & $\checkmark$ &   \\ \hline
Downlink & $\checkmark$ &  & $\checkmark$ & $\checkmark$ & $\checkmark$ & $\checkmark$ &   &  & $\checkmark$ \\ \hline
Statistical CSI & $\checkmark$ & &  &  &   &   &   & $\checkmark$ &   \\ \hline
Instantaneous CSI &   & $\checkmark$ & $\checkmark$ &  $\checkmark$ & $\checkmark$ & $\checkmark$ & $\checkmark$ &  $\checkmark$ & $\checkmark$ \\ \hline
Spectral efficiency & $\checkmark$ & $\checkmark$ & $\checkmark$ & $\checkmark$  & $\checkmark$ & $\checkmark$ & $\checkmark$ &   & \\ \hline
Energy efficiency & & & & &   &   &   & $\checkmark$ & $\checkmark$ \\ \hline
Far field & $\checkmark$ & $\checkmark$ & $\checkmark$ & $\checkmark$  &   &   &   & $\checkmark$ & $\checkmark$ \\ \hline
Near field & & &  &  & $\checkmark$ & $\checkmark$  & $\checkmark$ &   &  \\ \hline
\end{tabular}}
\end{table*}

\begin{table}[!h]
\caption{Meanings of notations.}
\label{table2}
\centering
{\color{black}
\begin{tabular}{|>{\raggedright\arraybackslash}p{2.9cm}|>{\raggedright\arraybackslash}p{4.7cm}|}
\hline
\textbf{Notations} & \textbf{Meanings} \\ \hline
$\mathbbmss{C}$ &  complex domain \\ \hline
$\jmath$ & imaginary unit with $\jmath^2=-1$ \\ \hline
Boldface lower-case variable ($\mathbf{x}$) & column vector \\ \hline
Boldface upper-case variable ($\mathbf{X}$) & matrix \\ \hline
$\mathbf{X}\geq 0$ &  positive semi-definite matrix \\ \hline
$[\mathbf{X}]_{i,j}$ & element in the $i$-th row and $j$-th column \\ \hline
$[\mathbf{x}]_i$ & $i$-th element of $[\mathbf{x}]_i$ \\ \hline $\|\mathbf{x}\|$ & Euclidean norm \\ \hline
$\mathbf{I}_N$ & $N \times N$ unit matrix \\ \hline
$\mathbf{0}_N$ & $N \times 1$ zeros column vector \\ \hline
$(\cdot)^{\mathsf T}$ & transpose \\ \hline
$(\cdot)^{\mathsf H}$ & conjugate transpose \\ \hline
$(\cdot)^{*}$ & complex conjugate \\ \hline
$\mathbbmss{E}\{\cdot\}$ & expectation operator \\ \hline
$\mathsf{diag}(\cdot)$ & diagonal elements extraction operator \\ \hline
$\mathsf{arg}(\cdot)$ & phase extraction operator \\ \hline
$[a, b]$ & closed interval from $a$ to $b$ \\ \hline
$\lceil \cdot \rceil$ & upward rounding operator \\ \hline
$\lfloor \cdot \rfloor$ & rounding down operator \\ \hline
${\rm{mod}}(\cdot,\cdot)$ & modulus operation \\ \hline
$\odot$ & Hadamard product operation \\ \hline
$\otimes$ & Kronecker product operation \\ \hline
$\Re(\cdot)$ & operation of taking the real part \\ \hline
\end{tabular}}
\end{table}

\begin{figure}[!t]
\centering
\setlength{\abovecaptionskip}{0pt}
\includegraphics[height=0.3\textwidth]{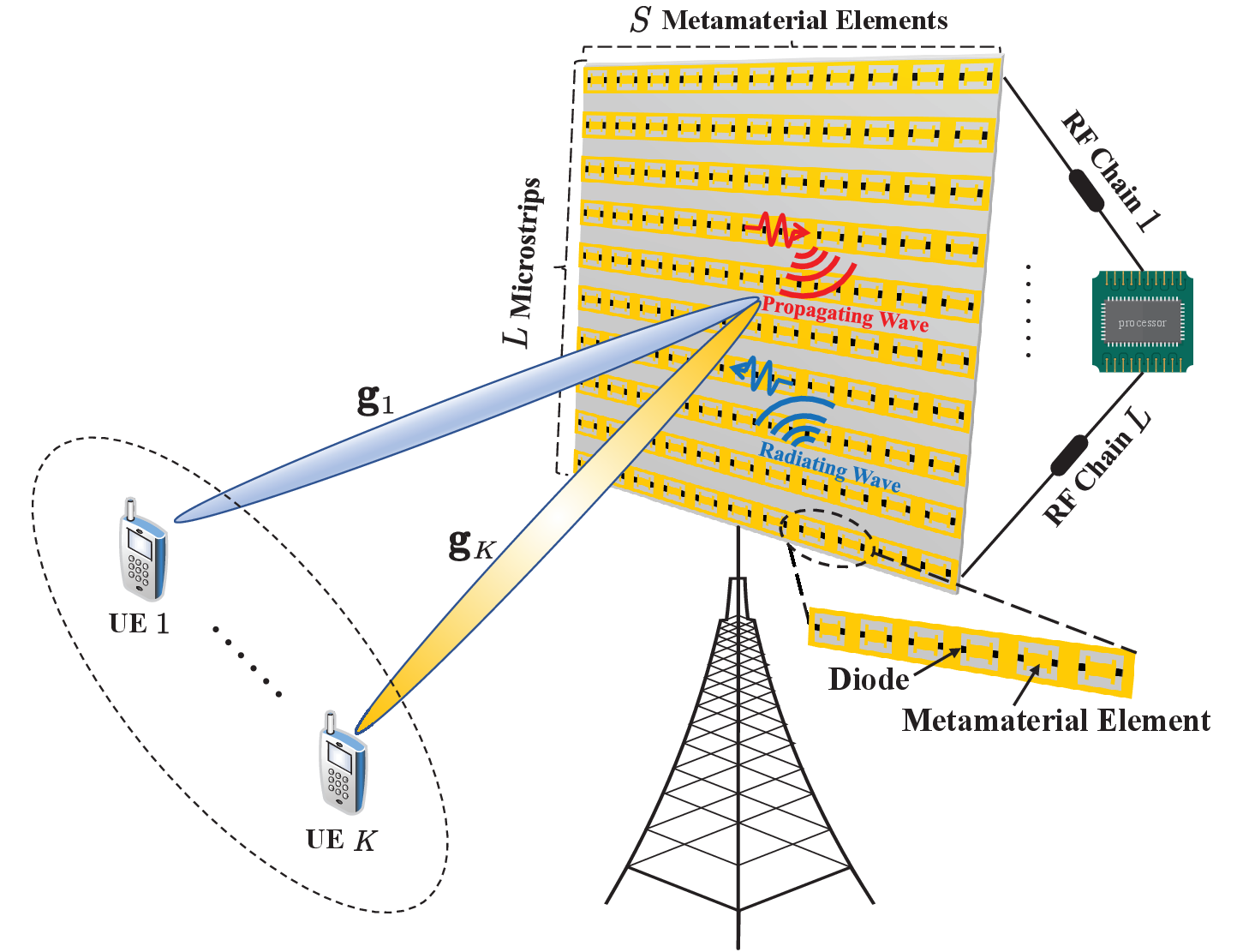}
\caption{Multiuser MISO system enabled with DMAs.}
\label{system_model}
\vspace{-15pt}
\end{figure}

\section{System Model}   \label{Section2}
We consider a DMA-enabled multiuser MISO communication system as shown in {\figurename} {\ref{system_model}}, where a DMA-aided BS communicates with $K$ single-antenna user equipments (UEs). In this section, we first provide an illustration of the DMA architecture. Then, we present the considered channel model, including uplink scenario and downlink scenario.

\subsection{DMA Architecture} \label{DMA_architecture}
As illustrated in {\figurename} {\ref{system_model}}, DMA is a planar array antenna composed of metamaterial elements, which integrate tunable components such as PIN diodes, varactors, and micro electro mechanical systems (MEMS). These metamaterial components are typically placed on waveguides that are connected to digital processors through input/output ports on RF chains. By independently adjusting the state of each tunable component through control circuits, the resonance characteristics of metamaterial elements can be changed, further achieving dynamic control of electromagnetic waves.

This paper assumes that DMA consists of $L$ one-dimensional waveguides based on microstrip, with each microstrip connected to an RF chain. We consider that $S$ metamaterial elements spaced at sub-wavelength distances are placed on each microstrip to transmit signals and receive waveforms. Therefore, the entire planar array contains $N=SL$ metamaterial elements. For uplink transmission, each metamaterial element transmits the observed propagation waveform along the waveguide to the corresponding RF chain, and the baseband signal processor receives a linear combination of these waveforms. The process is reversed for downlink transmission. Specifically, the signal behavior on DMA can be divided into two steps, including the frequency response of metamaterial components and the propagation of signals in microstrip. Below, we will provide a detailed explanation and modeling of these two physical phenomena.
\subsubsection{Frequency Response of Metamaterial Components}
In narrowband systems, prior research has indicated that the components can be assumed to have a flat frequency response, implying that the gain provided by metamaterial components remains consistent over the frequency range of interest. \cite{Shlezinger2021}. Generally, various tunable hardware implementations lead to different phase constraints, resulting in varying performance outcomes. {\color{black}Previous studies suggest that DMAs with a Lorentzian phase constraint, which combines amplitude and phase shift design, typically exhibit superior performance compared to those restricted to amplitude or binary amplitude constraints \cite{You2023,Zhang2024}. Thus, we assume that the DMA phase shifts follow Lorentzian constraint here, which can be implemented using varactor diodes.} Specifically, we denote $\mathbf{Q}\in \mathbbmss{C}^{N\times L}$ as the configurable weight matrix of DMAs, which is formulated as:
\begin{align}  \label{cons_Q}
[\mathbf{Q}]_{\left(l_1-1\right) S+s, l_2}=\left\{\begin{array}{ll}
q_{l_1, s}, & l_1=l_2 \\
0, & l_1 \neq l_2
\end{array},\right.
\end{align}
where $l_1,l_2\in\mathcal{L}=\{1,\dots,L\}$, $s \in \mathcal{S}=\{1,\dots,S\}$, and $q_{l,s}$ denotes the weight of the $s$-th metamaterial element in the $l$-th microstrip. We represent $\mathcal{Q}$ as the possible set of reconfigurable weight, which satisfies the Lorentzian constraint and can be expressed as
\begin{align}\label{cons_Q2}
\{q_{l,s}\}\in \mathcal{Q}=\left\{q=\frac{\jmath+ {\rm{e}}^{\jmath\theta}}{2}:\theta\in[0,2\pi]\right\}.
\end{align}

\subsubsection{Propagation Inside The Microstrip}
When the signal propagates through microstrips, it is mainly influenced by the wave number and element position, resulting in different phase shifts in different metamaterial components. In addition, the signal is also affected by attenuation caused by the characteristics and size of the waveguide material. Considering the above two factors, the propagation effect of signal propagation in waveguides can be modeled as:
\begin{align}\label{h}
h_{l, s}=e^{-\rho_{l, s}\left(\alpha_l+\jmath \gamma_l\right)}, \forall{l}\in \mathcal{L},\forall{s}\in \mathcal{S},
\end{align}
where $\alpha_l$ is the waveguide attenuation coefficient, $\gamma_l$ is the wavenumber, and $\rho_{l,s}$ is the distance between the $s$-th element and the RF port in the $l$-th microstrip.

In summary, DMA can be modeled as a linear process influenced by these two physical phenomena. For baseband-processed input signal $\mathbf{x}\in\mathbbmss{C}^{L \times 1}$, the output signal radiated into wireless environment through the DMA is given by:
\begin{align}\label{output_DMA}
\mathbf{y}_\text{o}=\mathbf{HQx},
\end{align}
where $\mathbf{H}\in \mathbbmss{C}^{N \times N}$ denotes the microstrip propagation coefficients which is a diagonal matrix with $(\mathbf{H})_{(l-1)S+s,(l-1)S+s}=h_{l,s}$ for $l \in \mathcal{L}$ and $s \in \mathcal{S}$.
\subsection{Uplink Channel Model}
We denote $x_k^\text{up}\in \mathbbmss{C}$ as the data symbol transmitted by UE $k$ with zero mean and unit variance, satisfying ${\mathbbmss{E}}\{x_k^\text{up} (x^{\text{up}}_{k^{\prime}})^{*}\}=0$, $\forall k\neq k^{\prime}$. Then the received signal at the BS's DMA elements is given by
{\setlength\abovedisplayskip{5pt}
\setlength\belowdisplayskip{5pt}
\begin{align}\label{y_BS}
\mathbf{y}_{\text{BS}}=\sum_{k=1}^{K}\mathbf{g}_k x_k^\text{up}+\mathbf{n},
\end{align}
}where $\mathbf{g}_k\in \mathbbmss{C}^{N \times 1}$ is the channel from the UE $k$ to the BS and $\mathbf{n}\sim \mathcal{CN}(\mathbf{0}_N,N_0 \mathbf{I}_N)$ is the additive white Gaussian noise (AWGN) with $N_0$ being the noise power. Without losing generality, we use the Rician fading to model the channel $\mathbf{g}_k$, which can be written by:
\begin{align}\label{uplink}
\mathbf{g}_k=\sqrt{\alpha_k\frac{K_0}{1+K_0}}\bar{\mathbf{g}}_k+ \sqrt{\alpha_k\frac{1}{1+K_0}}\tilde{\mathbf{g}}_k,
\end{align}
where $\alpha_k$ denotes the large scale path-loss coefficient between the BS and the $k$-th UE, $K_0$ is the Rician factor, $\bar{\mathbf{g}}_k$ is the line-of-sight (LoS) component of the channel from the UE $k$ to the BS, and $\tilde{\mathbf{g}}_k$ is the non-LoS (NLoS) component, respectively. {\color{black}Specifically, we assume that the LoS component is deterministic and fully characterized by the elevation and azimuth angles. In contrast, the NLoS component is treated as an unknown stochastic vector, modeled using the widely adopted Kronecker model \cite{Ying2014,Bjornson2017}. That is, $\tilde{\mathbf{g}}_k=\mathbf{R}_k^{1/2}\hat{\mathbf{g}}_k$, where $\mathbf{R}_k$ is the deterministic spatial correlation matrix and $\hat{\mathbf{g}}_k$ is a complex Gaussian random vector whose elements are independent and identically distributed (i.i.d.) complex random variables with zero mean and unit variance, i.e., $\hat{\mathbf{g}}_k\sim \mathcal{CN}(\mathbf{0},\mathbf{I}_N)$.} Based on the DMA structure characteristics described in the previous section, the received signal at RF ports in the uplink system can be expressed as:
\begin{align}\label{uplink}
\mathbf{y}_\text{up}=\mathbf{Q}^{\mathsf{H}}\mathbf{H}^{\mathsf{H}}\sum_{k=1}^{K}\mathbf{g}_k x_k^\text{up}+ \mathbf{Q}^{\mathsf{H}}\mathbf{H}^{\mathsf{H}}\mathbf{n} \in \mathbbmss{C}^{L \times 1}.
\end{align}
{\color{black}Specifically, when mutual coupling is taken into account, the received signal model in \eqref{y_BS} becomes $\mathbf{y}_{\text{BS}}=\mathbf{C}\sum_{k=1}^{K}\mathbf{g}_k x_k^\text{up}+\mathbf{n}$, where $\mathbf{C}$ denotes the coupling matrix \cite{Balanis2015}. By treating $\mathbf{C}\mathbf{g}_k$ as the effective channel for user $k$, the proposed analysis and optimization framework remains applicable. Therefore, mutual coupling is omitted in this paper for simplicity.}
\subsection{Downlink Channel Model}
Similar to the uplink case, we first denote $x_k^\text{d} \in \mathbbmss{C}$ as the information symbol transmitted by the BS and intended for the $k$-th UE, which also satisfies ${\mathbbmss{E}}\{x_k^\text{d} (x^{\text{d}}_{k^{\prime}})^{*}\}=0$ for $\forall k\neq k^{\prime}$. Based on the DMA model described in \eqref{output_DMA}, the received signal at UE $k$ is given by:
\begin{align}\label{downlink}
y_k=\mathbf{g}_k^{\mathsf{H}}\mathbf{H}\mathbf{Q}\mathbf{W} \mathbf{x}^\text{d}+ n_k,
\end{align}
where $\mathbf{x}^\text{d}=\{x_1^\text{d},\dots,x_K^\text{d}\}^{\mathsf T} \in {\mathbbmss{C}}^{K\times 1}$ denotes the transmit symbol vector, $\mathbf{W}=\{\mathbf{w}_1,\dots,\mathbf{w}_K\} \in {\mathbbmss{C}}^{L\times K}$ denotes the digital precoding matrix with ${\mathbf w}_k\in {\mathbbmss{C}}^{L\times 1}$ denotes the digital precoding vector for the UE $k$, and $n_k \sim \mathcal{CN}(0,N_k)$ is the AWGN at the receiver of UE $k$ with $N_k$ being the noise power.
\section{Uplink transmission} \label{Section3}
We start with uplink transmission and consider two classical decoding rules, i.e., MMSE-SIC and MMSE-nSIC. The MMSE-SIC decoding can reach the uplink sum rate capacity while the MMSE-nSIC decoding yields lower computational complexity. For both cases, we present the analytical ergodic sum rate expressions, establish SE optimization problem, and provide efficient algorithms to optimize the configurable weight matrix of the DMA.
\vspace{-10pt}
\subsection{MMSE-SIC} \label{Section_A}
\subsubsection{SE Analysis}
For the multiuser MISO system with MMSE-SIC decoding, the ergodic sum rate exploiting only statistical CSI can be expressed by \cite{Soysal2009}
{\color{black}
\begin{align}\label{R_SIC}
\mathcal{R}_\text{sic}=\mathbbmss{E}_{\mathbf{g}_k}\left\{\log\det\left( \mathbf{I}_L+\mathbf{Q}^{\mathsf H} \mathbf{H}^{\mathsf H} \sum_{k=1}^{K}\mathbf{g}_k \mathbf{g}_k^{\mathsf H} \mathbf{H}\mathbf{Q}(\mathbf{P})^{-1} \right)\right\},
\end{align}
}where $\mathbf{P}=N_0\mathbf{Q}^{\mathsf H}\mathbf{H}^{\mathsf H}\mathbf{H}\mathbf{Q}$ is the noise covariance matrix. The above equation with mathematical expectations is not a closed form expression, making it difficult to calculate. Therefore, to simplify the analysis and optimization, we provide the following theorem.
\begin{theorem} \label{Theorem_1}
The ergodic uplink sum rate with MMSE-SIC decoding is upper bounded by:
\begin{align}\label{R_SIC_upperbound}
\overline{\mathcal{R}}_\text{sic}=\log\det\left( \mathbf{I}_L+\mathbf{Q}^{\mathsf H} \mathbf{H}^{\mathsf H} \mathbf{G} \mathbf{G}^{\mathsf H} \mathbf{H}\mathbf{Q}(\mathbf{P})^{-1} \right),
\end{align}
where $\mathbf{G}=\left[\tilde{\mathbf{G}}_1,\dots,\tilde{\mathbf{G}}_K \right] \in \mathbbmss{C}^{N \times K(N+1)}$ with $\tilde{\mathbf{G}}_k=\left[\sqrt{\frac{\alpha_k K_0}{1+K_0}}\bar{\mathbf{g}}_k,\sqrt{\frac{\alpha_k }{1+K_0}}\mathbf{R}_k^{1/2}\right]\in \mathbbmss{C}^{N \times (N+1)}$.
\end{theorem}
\begin{IEEEproof}
For more details, please refer to Appendix \ref{Proof_A}.
\end{IEEEproof}
\begin{remark} \label{remark_1}
Theorem \ref{Theorem_1} provides a closed-form upper bound expression for $\mathcal{R}_\text{sic}$, which is more tractable to analyze and guide the system design. It is observed that the ergodic sum rate is depend on the DMA weight matrix $\mathbf{Q}$ with microstrip propagation coefficients $\mathbf{H}$ and the statistical CSI fixed. Thus, we could maximize $\overline{\mathcal{R}}_\text{sic}$ as a surrogate objective by optimizing $\mathbf{Q}$ to improve the ergodic sum rate.
\end{remark}
\subsubsection{Problem Formulation and Algorithm Design}
Base on Remark \ref{remark_1}, we aim to maximize $\overline{\mathcal R}_\text{sic}$ by optimizing DMA weight matrix $\mathbf{Q}$, which can be cast as
\begin{subequations}\label{P_sic}
\begin{align}
{\mathcal{P}}_{\text{sic}}:~\max_\mathbf{Q}&~\overline{\mathcal R}_\text{sic}\label{P_X_Obj}\\
{\text{s.t.}}&~\eqref{cons_Q},\eqref{cons_Q2}.
\end{align}
\end{subequations}
The problem $\mathcal{P}_\text{sic}$ is difficult to solve due to the non-convex objective function corresponding $\mathbf{Q}$ and non-convex constraint of $\mathbf{Q}$ in \eqref{cons_Q},\eqref{cons_Q2}. Therefore, we first resort to transform the original problem to a more tractable form. Specifically, we use the classic WMMSE method \cite{Shi2011} by treating $\overline{\mathcal{R}}_\text{sic}$ as the capacity of a equivalent MIMO channel $\mathbf{y}_\text{e}=\mathbf{Q}^{\mathsf H}\mathbf{H}^{\mathsf H} \left(\mathbf{G}\mathbf{s}_\text{e} +\mathbf{z}_\text{e}\right)$, where $\mathbf{s}_\text{e} \sim \mathcal{CN}(\mathbf{0}_{K(N+1)},\mathbf{I}_{K(N+1)})$ and $\mathbf{z}_\text{e} \sim \mathcal{CN}(\mathbf{0}_N,N_0 \mathbf{I}_N)$ denote the transmit signal and AWGN, respectively. By denoting $\mathbf{U}_\text{e}\in \mathbbmss{C}^{L \times K(N+1)}$ as the linear receive matrix, the mean square error matrix is given by
\begin{align}
\mathbf{E}_\text{e}=\mathbbmss{E}_{\mathbf{s}_\text{e}, \mathbf{n}_\text{e}}\left\{\left(\mathbf{U}_\text{e}^{\mathsf H} \mathbf{y}_\text{e}-\mathbf{s}_\text{e}\right) \left(\mathbf{U}_\text{e}^{\mathsf H} \mathbf{y}_\text{e}-\mathbf{s}_\text{e}\right)^{\mathsf H}\right\}.
\end{align}
Then, $\mathcal{P}_\text{sic}$ can be equivalently transformed to the following problem \cite{Shi2011}:
\begin{subequations}\label{P1}
\begin{align}
{\mathcal P}_1: ~\min_{\mathbf{Q},\mathbf{U}_\text{e},\mathbf{W}_\text{e}} &~ \mathsf{tr} \left(\mathbf{W}_\text{e}\mathbf{E}_\text{e}\right)- \log\det\left(\mathbf{W}_\text{e}\right) \label{P1_obj}\\
{\rm{s.t.}}&~ \eqref{cons_Q},\eqref{cons_Q2},
\end{align}
\end{subequations}
where $\mathbf{W}_\text{e} \in \mathbbmss{C}^{K(N+1) \times K(N+1)}$ is a semidefinite auxiliary matrix. We note that the objective function of the WMMSE problem $\mathcal{P}_1$ is convex over each optimization variable while holding others fixed, which means that $\mathcal{P}_1$ is more tractable than the problem $\mathcal{P}_\text{sic}$. On this basis, we could alternately optimize $\mathbf{W}_\text{e}$, $\mathbf{U}_\text{e}$, and $\mathbf{Q}$ by keeping other variables fixed.

\textbf{Step 1}, we optimize the auxiliary variable $\mathbf{W}_\text{e}$ by checking its first-order optimality condition. Specifically, by fixing other variables, the optimal $\mathbf{W}_\text{e}$ is given by
\begin{align}\label{W_e}
\mathbf{W}_\text{e}^{\star}=\mathbf{E}_\text{e}^{-1}.
\end{align}

\textbf{Step 2}, we optimize the receiving matrix $\mathbf{U}_\text{e}$ by keeping other variables fixed. Following the method in Step 1, we obtain the optimal $\mathbf{U}_\text{e}$ as follows
\begin{align}\label{U_e}
\mathbf{U}_\text{e}^\star=\left(\mathbf{Q}^{\mathsf H} \mathbf{H}^{\mathsf H}\mathbf{G}\mathbf{G}^{\mathsf H} \mathbf{H}\mathbf{Q}+\mathbf{P}\right)^{-1} \mathbf{Q}^{\mathsf H} \mathbf{H}^{\mathsf H}\mathbf{G}.
\end{align}

\textbf{Step 3}, we optimize the DMA weight matrix $\mathbf{Q}$ with $\mathbf{W}_\text{e}$ and $\mathbf{U}_\text{e}$ fixed. Let $\mathbf{h}=\mathsf{diag}(\mathbf{H})$ and $\mathbf{q}=\{q_1,\dots,q_N\}\in \mathbbmss{C}^{N \times 1}$ denote the DMA weight vector with $q_i=[\mathbf{Q}]_{i,\lceil \frac{i}{S}\rceil}$ for $i \in \mathcal{N}=\{1,\dots,N\}$, then we rewrite $\mathbf{HQ}$ as $\tilde{\mathbf{Q}}\tilde{\mathbf{H}}$, where $\tilde{\mathbf{Q}}\in \mathbbmss{C}^{N \times N}$ is a diagonal matrix with $[\tilde{\mathbf{Q}}]_{i,i}=q_i$ and $\tilde{\mathbf{H}}\in \mathbbmss{C}^{N \times L}$ is a block diagonal matrix which is given by
\begin{align}\label{tilde_H}
[\tilde{\mathbf{H}}]_{n, l}=\left\{\begin{array}{ll}
[\mathbf{h}]_n, & l=\lceil\frac{n}{S}\rceil \\
0, & \text{otherwise}
\end{array},\right.
\end{align}
for $ n\in\mathcal{N}$ and $\forall l \in \mathcal{L}$. Ignoring irrelevant terms, we obtain an optimization problem as follows
{\setlength\abovedisplayskip{2pt}
\setlength\belowdisplayskip{2pt}
\begin{subequations}\label{P2}
\begin{align}
{\mathcal P}_{\tilde{\mathbf{Q}}}: ~\min_{\tilde{\mathbf{Q}}} &~ \mathsf{tr} \big({\tilde{\mathbf{Q}}}^{\mathsf H} \mathbf{A}_0 {\tilde{\mathbf{Q}}} \mathbf{B}_0-{\tilde{\mathbf{Q}}} \mathbf{C}_0-\mathbf{C}_0^{\mathsf H}{\tilde{\mathbf{Q}}}^{\mathsf H} \big) \\
{\rm{s.t.}}&~ [{\tilde{\mathbf{Q}}}]_{n,n} \in \mathcal{Q}, \forall n \in \mathcal{N},
\end{align}
\end{subequations}
}where $\mathbf{A}_0=\mathbf{G}\mathbf{G}^{\mathsf H}+N_0 \mathbf{I}_N$, $\mathbf{B}_0={\tilde{\mathbf{H}}}\mathbf{U}_\text{e} \mathbf{W}_\text{e}\mathbf{U}_\text{e}^{\mathsf H} {\tilde{\mathbf{H}}}^{\mathsf H}$, and $\mathbf{C}_0={\tilde{\mathbf{H}}}\mathbf{U}_\text{e} \mathbf{W}_\text{e}\mathbf{G}^{\mathsf H}$. By denoting $\mathbf{D}_0=\mathbf{A}_0 \odot \mathbf{B}_0^{\mathsf T}$ and $\mathbf{c}_0=\mathsf{diag}(\mathbf{C}_0^{\mathsf H})$, $\mathcal{P}_{\tilde{\mathbf{Q}}}$ can be further transformed to the following problem regarding the DMA weights vector $\mathbf{q}$:
\begin{subequations}\label{Pq_sic}
\begin{align}
{\mathcal P}_{\mathbf{q}}^{\text{sic}}: ~\min_{\mathbf{q}} &~ f_{\mathbf{q}}(\mathbf{q})=\mathbf{q}^{\mathsf H} \mathbf{D}_0 \mathbf{q} - 2\Re\left(\mathbf{q}^{\mathsf H}\mathbf{c}_0\right) \label{Pq_sic_obj}\\
{\rm{s.t.}}&~ q_n \in \mathcal{Q}, \forall n \in \mathcal{N}. \label{Pq_sic_cons}
\end{align}
\end{subequations}
Due to the constraints acting on each element, we propose an element-wise refinement (EWR) algorithm to alternately optimize the elements of $\mathbf{q}$. Specifically, the objective function in \eqref{Pq_sic_obj} can be rewritten as
\begin{align}
f_{\mathbf{q}}(\mathbf{q})=\sum_{i=1}^{N}\sum_{j=1}^{N}q_i^{*} [\mathbf{D}_0]_{i,j}q_j-2\Re\Big(\sum_{i=1}^{N}q_i^{*} [\mathbf{c}_0]_i \Big).
\end{align}
By fixing $\{q_i\}_{i\neq n}$ and omitting the terms independent with $q_n$, we obtain the subproblem as follows:
\begin{subequations}\label{P_n}
\begin{align}
{\mathcal P}_{q_n}: ~\min_{q_n} &~ f_n(q_n) = 2\Re\Big(\sum_{m\neq n}^{N} q_n^{*} [\mathbf{D_0}]_{n,m}q_m \Big) \label{Pn_obj}\\
 &~ -2\Re\left(q_n^{*}[\mathbf{c}_0]_n\right) + q_n^{*}q_n[\mathbf{D_0}]_{n,n} \nonumber \\
{\rm{s.t.}}&~  q_n \in \mathcal{Q}, \forall n \in \mathcal{N}.
\end{align}
\end{subequations}
This problem is intractable due to the non-convex constraint. However, we observe that the problem of optimizing $q_n$ is actually optimizing the phase $\theta_n$ from the Lorentzian constraint $q_n=\frac{\jmath+ {\rm{e}}^{\jmath\theta_n}}{2}$. Thus, by instituting $q_n=\frac{\jmath+ {\rm{e}}^{\jmath\theta_n}}{2}$ to $f_n(q_n)$ in \eqref{Pn_obj} and ignoring constant terms, $\mathcal{P}_{q_n}$ can be transformed to the following problem:
\begin{subequations}
\begin{align}
{\mathcal P}_{\theta_n}: ~\min_{\theta_n} &~ g_n(\theta_n) = \Re\left( \rm{e}^{-\jmath \theta_n} (-\eta_n)  \right)\\
{\rm{s.t.}}&~ \theta_n \in [0, 2\pi],
\end{align}
\end{subequations}
where $\eta_n=\left([\mathbf{c}_0]_n- \sum_{m\neq n}^{N} [\mathbf{D_0}]_{n,m}q_m- \frac{\jmath}{2}[\mathbf{D_0}]_{n,n} \right)$. Therefore, the optimal $\theta_n$ is given by
\begin{align} \label{qn_3}
\theta_n^{\star} = \argmin\Re\left( \rm{e}^{-\jmath \theta_n} (-\eta_n) \right)=\angle \eta_n,
\end{align}
which yields the optimal $q_n^{\star}=\frac{\jmath+ {\rm{e}}^{\jmath\theta_n^{\star}}}{2}$.

\begin{algorithm}[!t]
    \caption{EWR-based method}
    \label{Algorithm1}
    Initialize $\mathbf{q}^{(t)}=[q_1^{(t)}, \dots, q_N^{(t)}]$ and  iteration index $t=0$\;
    \Repeat{convergence}
    {
    \For{$n=1:N$}
    {Optimize the phase shift $q_n^{(t)}$ based on \eqref{qn_3}\;
    }
    Set $t=t+1$\;
    }
    Calculate the DMA weight matrix $\mathbf{Q}$ based on  $\mathbf{q}^{(t)}$.
\end{algorithm}

\begin{algorithm}[!t]
    \caption{WMMSE-based method}
    \label{Algorithm2}
    Initialize $\left\{\mathbf{W}_\text{e}^{(j)}, \mathbf{U}_\text{e}^{(j)}, \mathbf{Q}^{(j)}\right\}$ and  iteration index $j=0$\;
    \Repeat{convergence}
    {
    Optimize $\mathbf{W}_\text{e}^{(j+1)}$ based on \eqref{W_e}\;
    Optimize $\mathbf{U}_\text{e}^{(j+1)}$ based on \eqref{U_e}\;
    Optimize $\mathbf{Q}^{(j+1)}$ based on the Algorithm \ref{Algorithm1}\;
    Set $j=j+1$\;
    }
\end{algorithm}

The proposed EWR-based method is summarized in Algorithm \ref{Algorithm1}. It is obvious that the objective function in \eqref{Pq_sic_obj} monotonically decreases, which ensures that Algorithm \ref{Algorithm1} converges to a local optimal solution of ${\mathcal P} _{\mathbf{q}} ^{\text{sic}}$. The
computational complexity of Algorithm \ref{Algorithm1} can be estimated as $\mathcal{O}(I_\text{EWR}^\text{sic} N^2)$, where $I_\text{EWR}^\text{sic}$ represents the number of iterations for the case of MMSE-SIC decoding.

\subsubsection{Convergence and Complexity Analysis}
By alternately optimizing $\{\mathbf{W}_{\text{e}}, \mathbf{U}_{\text{e}}, \mathbf{Q}\}$ until the objective function in \eqref{P1_obj} converges, we obtain a local optimal solution of $\mathcal{P}_\text{sic}$. The proposed WMMSE-based method is summarized in Algorithm \ref{Algorithm2}. As mentioned before, each step of the WMMSE-based approach, i.e., the optimization of $\{\mathbf{W}_{\text{e}}, \mathbf{U}_{\text{e}}, \mathbf{Q}\}$, will not increase the objective function of $\mathcal{P}_1$, which guarantees the convergence of the Algorithm \ref{Algorithm2}. The computational complexity can be estimated by $I_{\text{W}}^\text{sic}\left(3K^3(N+1)^3+ I_{\text{EWR}}^\text{sic}N^2\right)$, where $I_{\text{W}}^\text{sic}$ is the iteration number of the WMMSE-based method.
\vspace{-15pt}
\subsection{MMSE-nSIC}
\subsubsection{SE Analysis}
The ergodic sum rate of the considered uplink system with MMSE-nSIC decoding is given by
{\color{black}
\begin{align}\label{MMSE-nSIC}
\mathcal{R}_{\text{nsic}}&=\mathbbmss{E}_{\mathbf{g}_k} \Bigg\{\sum_{k=1}^{K} \log\det\Big(\mathbf{I}_L +\mathbf{Q}^{\mathsf H} \mathbf{H}^{\mathsf H}\mathbf{g}_k\mathbf{g}_k^{\mathsf H} \mathbf{HQ} \nonumber \\
& \big(\sum_{i\neq k}^{K}\mathbf{Q}^{\mathsf H} \mathbf{H}^{\mathsf H}\mathbf{g}_i\mathbf{g}_i^{\mathsf H} \mathbf{HQ}+\mathbf{P}\big)^{-1} \Big)\Bigg\}.
\end{align}}
Similar to the process in section \ref{Section_A}, we propose Theorem \ref{Theorem_2} to pursue an approximation.
\begin{theorem} \label{Theorem_2}
The ergodic sum rate with MMSE-nSIC decoding can be approximated by
\begin{align}\label{MMSE_nSIC_approximate}
\tilde{\mathcal{R}}_{\text{nsic}}& = \sum_{k=1}^{K} \log\det\Big(\mathbf{I}_L+\mathbf{Q}^{\mathsf H} \mathbf{H}^{\mathsf H}\tilde{\mathbf{G}}_k \tilde{\mathbf{G}}_k^{\mathsf H} \mathbf{HQ}  \nonumber \\
& \big(\sum_{i\neq k}^{K}\mathbf{Q}^{\mathsf H} \mathbf{H}^{\mathsf H}\tilde{\mathbf{G}}_i \tilde{\mathbf{G}}_i^{\mathsf H} \mathbf{HQ}+ \mathbf{P}\big)^{-1} \Big).
\end{align}
\end{theorem}
\begin{IEEEproof}
For more details, please refer to Appendix \ref{Proof_B}.
\end{IEEEproof}
\begin{remark} \label{remark_2}
Theorem \ref{Theorem_2} provides a closed-form approximation of the ergodic sum rate with MMSE-nSIC decoding. Similar to Remark \ref{remark_1}, we could optimize $\mathbf{Q}$ to maximize the approximate expression $\tilde{\mathcal{R}}_\text{sic}$ and thus improve the ergodic sum rate.
\end{remark}
\subsubsection{Problem Formulation and Algorithm Design}
According to Remark \ref{remark_2}, we can establish the following problem
\begin{subequations}\label{P_nsic}
\begin{align}
{\mathcal{P}}_{\text{nsic}}:~\max_\mathbf{Q}&~\tilde{\mathcal R}_\text{nsic} \\
{\text{s.t.}}&~ \eqref{cons_Q},\eqref{cons_Q2}.
\end{align}
\end{subequations}
It is noted that $\tilde{\mathcal{R}}_\text{nsic}$ can be considered as the channel capacity of a downlink MU-MIMO system where a BS equipped with $(N+1)$ antennas transmits signals to $K$ users enabled with DMAs. In this equivalent system, the signal received at the $k$-th user is given by
\begin{align}
\mathbf{y}_k= \mathbf{Q}^{\mathsf H} \mathbf{H}^{\mathsf H} \left(\sum_{i=1}^{K} \tilde{\mathbf{G}}_i \mathbf{s}_i+\mathbf{z}_k \right),
\end{align}
where $\mathbf{s}_i \sim \mathcal{CN} (\mathbf{0}_{N+1},\mathbf{I}_{N+1})$ is the information symbol for user $i$ and $\mathbf{z}_k\sim \mathcal{CN}(\mathbf{0}_N, N_0 \mathbf{I}_N)$ is the noise for user $k$. Similar to the case of MMSE-SIC decoding, the equivalent WMMSE problem can be expressed as
\begin{subequations}\label{P2}
\begin{align}
{\mathcal P}_2: ~\min_{\mathbf{Q},\mathbf{U}_k,\mathbf{W}_k} &~ \sum_{k=1}^{K}\left[\mathsf{tr} \left(\mathbf{W}_k\mathbf{E}_k\right)- \log\det\left(\mathbf{W}_k\right)\right] \label{P2_obj}\\
{\rm{s.t.}}&~ \eqref{cons_Q},\eqref{cons_Q2},
\end{align}
\end{subequations}
where $\mathbf{W}_k \in \mathbbmss{C}^{(N+1) \times (N+1)} \succeq 0$ is an auxiliary matrix, $\mathbf{E}_k=\mathbbmss{E}_{\mathbf{s}_k, \mathbf{z}_k}\left\{\left(\mathbf{U}_k^{\mathsf H} \mathbf{y}_k-\mathbf{s}_k\right) \left(\mathbf{U}_k^{\mathsf H} \mathbf{y}_k-\mathbf{s}_k\right)^{\mathsf H}\right\}$ with $\mathbf{U}_k\in \mathbbmss{C}^{L \times (N+1)}$ being the linear receive matrix corresponding to user $k$. Following the process described in section \ref{Section_A}, we optimize ${\mathbf{Q},\{\mathbf{U}_k\},\{\mathbf{W}_k}\}$ alternately until the objective function in \eqref{P2_obj} converges.

\textbf{Step 1}, we optimize $\mathbf{W}_k$ by checking its first-order optimality condition, which yields the optimal $\mathbf{W}_k$ as $\mathbf{W}_k^{\star}=\mathbf{E}_k^{-1}$.

\textbf{Step 2}, we optimize $\mathbf{U}_k$ with other variables fixed. Similar to Step 1, the optimal solution of $\mathbf{U}_k$ is given by $\mathbf{U}_k^{\star}=\left(\mathbf{Q}^{\mathsf H} \mathbf{H}^{\mathsf H}\mathbf{G}\mathbf{G}^{\mathsf H} \mathbf{H}\mathbf{Q}+ \mathbf{P}\right)^{-1} \mathbf{Q}^{\mathsf H} \mathbf{H}^{\mathsf H} \tilde{\mathbf{G}}_k$.

\textbf{Step 3}, we optimize $\mathbf{Q}$ with $\{\mathbf{W}_k\}$ and $\{\mathbf{U}_k\}$ fixed. By ignoring irrelevant terms and after some basic mathematical manipulations, the problem with respect to $\mathbf{q}$ is given by
\begin{subequations}\label{Pq_nsic}
\begin{align}
{\mathcal P}_{\mathbf{q}}^{\text{nsic}}: ~\min_{\mathbf{q}} &~ \mathbf{q}^{\mathsf H} \mathbf{D}_1 \mathbf{q} - 2\Re(\mathbf{q}^{\mathsf H}\mathbf{c}_1) \label{Pq_obj}\\
{\rm{s.t.}}&~ q_n \in \mathcal{Q}, \forall n \in \mathcal{N}, \label{Pq_cons}
\end{align}
\end{subequations}
where $\mathbf{D}_1=\mathbf{A}_0 \odot \mathbf{B}_1$ with $\mathbf{B}_1 = \sum_{k=1}^{K}\tilde{\mathbf{H}}\mathbf{U}_k\mathbf{W}_k \mathbf{U}_k^{\mathsf H}\tilde{\mathbf{H}}^{\mathsf H}$, $\mathbf{c}_1=\mathsf{diag}(\mathbf{C}_1)$ with $\mathbf{C}_1=\sum_{k=1}^{K}\tilde{\mathbf{H}} \mathbf{U}_k\mathbf{W}_k \mathbf{G}_k^{\mathsf H}$. It is observed that the problem ${\mathcal P}_{\mathbf{q}}^{\text{nsic}}$ has the same form as problem ${\mathcal P}_{\mathbf{q}}^{\text{sic}}$, which can be solved by the proposed EWR-based method.
\subsubsection{Convergence and Complexity Analysis}
The proposed WMMSE-based method is similar to the \ref{Section_A}, thus can ensure convergence to a local optimal solution. Moreover, the computational complexity scales with $I_{\text{W}}^\text{nsic}\left(3K^3(N+1)^3+ I_{\text{EWR}}^\text{nsic}N^2\right)$, where $I_{\text{W}}^\text{nsic}$ and $I_{\text{EWR}}^\text{nsic}$ are the iteration number of the WMMSE-based approach and the EWR-based method for MMSE-nSIC decoding, respectively.
\section{Downlink transmission}  \label{Section4}
In this section, we consider the downlink transmission. Similarly, we first focus on the SE analysis, then formulate a SE optimization problem and provide an efficient algorithm.
\vspace{-5pt}
\subsection{SE Analysis}
According to the system model in \eqref{downlink}, the ergodic sum rate of the downlink transmission is given by
{\color{black}
\begin{align} \label{R_d}
\mathcal{R}_\text{d}=\sum_{k=1}^{K}\mathbbmss{E}_{\mathbf{g}_k}\left\{ \log_2\left(1+\frac{|\mathbf{g}_k^{\mathsf H} \mathbf{H} \mathbf{Q}\mathbf{w}_k|^2}{\sum_{i\neq k}^{K}|\mathbf{g}_k^{\mathsf H} \mathbf{H} \mathbf{Q}\mathbf{w}_i|^2+N_k }\right)\right\}.
\end{align}
}

Due to the calculation of mathematical expectation, it is difficult to analyze $\mathcal{R}_\text{d}$ directly. {\color{black}To obtain a more tractable approximation of $\mathcal{R}_\text{d}$, we propose the Corollary \ref{corollary1} by employing a similar approach to that used in the derivation of Theorem \ref{Theorem_2}.
\begin{corollary} \label{corollary1}
The ergodic sum rate for downlink system can be approximated by
\begin{align}\label{R_d_approximate}
{\tilde{\mathcal{R}}}_{\text{d}}=\sum_{k=1}^{K} \log_2\left(1+\frac{\|\tilde{\mathbf{G}}_k^{\mathsf H} \mathbf{H} \mathbf{Q}\mathbf{w}_k\|^2}{\sum_{i\neq k}^{K}\|\tilde{\mathbf{G}}_k^{\mathsf H} \mathbf{H} \mathbf{Q}\mathbf{w}_i\|^2+N_k }\right).
\end{align}
\end{corollary}
}

Compared to $\mathcal{R}_\text{d}$, the closed-form approximation is more tractable to analyze and optimize. In the following, we will regard $\tilde{\mathcal{R}}_\text{d}$ as the objective to optimize the BS transmit beamformer $\mathbf{W}$ and the DMA weight matrix $\mathbf{Q}$.
\vspace{-5pt}
\subsection{Problem Formulation}
Based on the approximation derived below, we aim to jointly optimize the $\mathbf{W}$ and $\mathbf{Q}$ for maximizing $\tilde{\mathcal{R}}_\text{d}$ as well as improve SE. Thus, the problem can be formulated as
\begin{subequations}\label{P_d}
\begin{align}
{\mathcal{P}}_{\text{d}}:~\max_{\mathbf{Q},\mathbf{w}} &~\tilde{\mathcal R}_\text{d} \\
{\text{s.t.}}&~\eqref{cons_Q},\eqref{cons_Q2},\\ &~\sum\nolimits_{k=1}^{K}\|\mathbf{HQ}\mathbf{w}_k\|^2\leq P_\text{max}, \label{P_d_cons}
\end{align}
\end{subequations}
where $P_\text{max}$ is the maximum available transmit power of the BS. This problem is intractable due to the following reasons. First, the objective function is non-convex concerning $\mathbf{Q}$ and $\mathbf{W}$. Second, the structure constraint of $\mathbf{Q}$ is non-convex, which further complicates the optimization procedure. Third, the variables $\mathbf{Q}$ and $\mathbf{W}$ are coupled in the constraint \eqref{P_d_cons}, which makes it difficult to optimize them concurrently. Next, we resort to the FP frame \cite{Shen2018} to seek a more mathematically tractable problem by the following theorem.
\begin{theorem} \label{Theorem_4}
Problem $\mathcal{P}_\text{d}$ is equivalent to
\begin{subequations}\label{P3}
\begin{align}
{\mathcal P}_3: ~\max_{\mathbf{Q},\mathbf{W},\bm\rho,\bm\Gamma} &~ \mathcal{F}_1 (\mathbf{Q},\mathbf{W},\bm\rho,\bm\Gamma) =\sum_{k=1}^{K} \log_2(1+\rho_k) \nonumber\\
&~ -\sum_{k=1}^{K}\rho_k +\sum_{k=1}^{K}(1+\rho_k) \mathcal{A}_k \label{P3_obj}\\
{\rm{s.t.}}&~ \eqref{cons_Q},\eqref{cons_Q2}, \\
&~\sum\nolimits_{k=1}^{K}\|\mathbf{HQ}\mathbf{w}_k\|^2\leq P_\text{max}, \label{P3_cons}
\end{align}
\end{subequations}
where $\bm\rho=[\rho_1,\dots,\rho_K]$ and $\bm\Gamma=[\bm\gamma_1,\dots,\bm\gamma_K]$ are two series of auxiliary variables, and $\mathcal{A}_k=2\Re\left(\bm\gamma_k^{\mathsf H}\tilde{\mathbf{G}}_k^{\mathsf H} \mathbf{H} \mathbf{Q}\mathbf{w}_k \right) - \left(\sum_{i=1}^{K}\mathbf{w}_i^{\mathsf H} \mathbf{Q}^{\mathsf H}\mathbf{H}^{\mathsf H} \tilde{\mathbf{G}}_k \tilde{\mathbf{G}}_k^{\mathsf H} \mathbf{H} \mathbf{Q} \mathbf{w}_i+N_k\right)\bm\gamma_k^{\mathsf H}\bm\gamma_k$.
\end{theorem}
\begin{proof}
Please see Appendix \ref{Proof_D} for more details.
\end{proof}
%
Note that problem $\mathcal{P}_3$ is easier to solve since $\mathcal{F}_1$ becomes concave over each variable with others being fixed. However, the coupling of $\mathbf{Q}$ and $\mathbf{W}$ in the power constraint still poses challenges. To tackle this issue, some works relaxed the power constraint by removing H and Q in \eqref{P3_cons} \cite{Zhang2021,Zhang2022,Kimaryo2023}. Then, one can design the overall system by alternately optimizing variables $\{\mathbf{Q},\mathbf{W},\bm\rho,\bm\Gamma\}$ until the objective function in \eqref{P3_cons} converges. Finally, the digital precoder $\mathbf{W}$ is scaled to resatisfy the constraint in \eqref{P3_cons}. However, this approach is crude and may reduce the system performance. Therefore, to compensate for this deficiency, we propose a PDD-based joint optimization algorithm to solve $\mathcal{P}_3$ \cite{Shi2020}.

\subsection{The PDD-Based Algorithm}  \label{PDD}
The proposed PDD-based algorithm is characterized by an embedded double-loop structure. Specifically, the inner loop solves the augmented Lagrangian (AL) subproblem while the outer loop updates the penalty parameter or dual variable based on the constraint violation. To deal with the coupling constraint in \eqref{P3_cons}, we introduce a new auxiliary variable $\mathbf{V}$ that is subjected to $\mathbf{V}=\mathbf{HQW}$ and $\mathsf{tr}(\mathbf{V}\mathbf{V}^{\mathsf H}) \leq P_\text{max}$. Then $\mathcal{P}_3$ can be rewritten as
\begin{subequations}\label{P4}
\begin{align}
\mathcal{P}_4:~\max_{\mathbf{Q},\mathbf{W},\bm\rho,\bm\Gamma} &~ \sum_{k=1}^{K} \log_2(1+\rho_k) -\sum_{k=1}^{K}\rho_k +\sum_{k=1}^{K}(1+\rho_k)\mathcal{C}_k \\
{\rm{s.t.}}&~ \eqref{cons_Q},\eqref{cons_Q2},~ \mathsf{tr}\left(\mathbf{V}^\mathsf{H} \mathbf{V}\right)\leq P_\text{max},\\
&~\mathbf{V}=\mathbf{HQW},   \label{eq_V}
\end{align}
\end{subequations}
where $\mathbf{V}=\left\{\mathbf{v}_1,\dots,\mathbf{v}_K\right\}$ and $\mathcal{C}_k=2\Re\left(\bm\gamma_k^{\mathsf H}\tilde{\mathbf{G}}_k^{\mathsf H} \mathbf{v}_k \right) - \left(\sum_{i=1}^{K}\mathbf{v}_i^{\mathsf H} \tilde{\mathbf{G}}_k \tilde{\mathbf{G}}_k^{\mathsf H} \mathbf{v}_i+N_k\right)\bm\gamma_k^{\mathsf H}\bm\gamma_k$. By moving the equality constraints as a penalty term to the objective function, we convert problem $\mathcal{P}_4$ into its AL form:
\begin{subequations}\label{P_5}
\begin{align}
\mathcal{P}_5:~\max_{\mathbf{Q},\mathbf{W},\bm\rho,\bm\Gamma} &~ \sum_{k=1}^{K} \log_2(1+\rho_k) -\sum_{k=1}^{K}\rho_k +\sum_{k=1}^{K}(1+\rho_k) \mathcal{C}_k  \nonumber\\
&~  -\frac{1}{2\beta}\left\|\mathbf{HQW}-\mathbf{V}+\beta{\bm\Xi}\right\| _\text{F}^2 \label{P_AL_obj}\\
{\rm{s.t.}}&~ \eqref{cons_Q},\eqref{cons_Q2}, ~\mathsf{tr}\left(\mathbf{V}^\mathsf{H} \mathbf{V}\right)\leq P_\text{max},
\end{align}
\end{subequations}
where $\beta >0$ is the penalty parameter, $\bm\Xi=\left\{\bm\xi_1,\dots, \bm\xi_K\right\} \in \mathbbmss{C}^{N \times K}$ denotes the dual variables associated with the equality constraint in \eqref{eq_V}. In the following, we solve the AL subproblem in the inner loop by alternately optimizing $\{\bm\rho,\bm\Gamma,\mathbf{V},\mathbf{W},\mathbf{Q}\}$ one by one with other variables fixed.

\textbf{Step 1}, we optimize $\bm\rho$ with other variables fixed. By checking the first-order optimality condition, the optimal $\rho_k^\star$ is given by
\begin{align}\label{opt_rho}
\rho_k^{\star}=\frac{\mathbf{v}_k^{\mathsf H}\tilde{\mathbf{G}}_k \tilde{\mathbf{G}}_k^{\mathsf H}\mathbf{v}_k}{\sum_{i\neq k}^{K} \mathbf{v}_i^{\mathsf H}\tilde{\mathbf{G}}_k \tilde{\mathbf{G}}_k^{\mathsf H}\mathbf{v}_i+N_k}.
\end{align}

\textbf{Step 2}, we optimize $\bm\Gamma$ by holding other variables fixed. Similar to Step 1, we check the first-order optimality condition of $\bm\gamma_k$ and obtain the optimal solution of $\bm\gamma_k$ written by
\begin{align}\label{opt_Gamma}
\bm\gamma_k^{\star}=\left(\sum_{i=1}^{K}\mathbf{v}_i^{\mathsf H}\tilde{\mathbf{G}}_k \tilde{\mathbf{G}}_k^{\mathsf H}\mathbf{v}_i+N_k\right)^{-1} \tilde{\mathbf{G}}_k^{\mathsf H} \mathbf{v}_k.
\end{align}

\textbf{Step 3}, we optimize $\mathbf{V}$ by fixing other variables. Specifically, by ignoring irrelevant items, the subproblem for optimizing $\mathbf{v}_k$ can be written as
\begin{subequations}\label{P_V}
\begin{align}
{\mathcal P}_{\mathbf{v}_k}: ~\min_{\mathbf{v}_k} &~
f_{\mathbf{V}}(\mathbf{v}_k)=\mathbf{v}_k^{\mathsf H} \bm\Omega \mathbf{v}_k-2\Re \left( (1+\rho_k) \mathbf{v}_k^{\mathsf H}   \tilde{\mathbf{G}}_k \bm\gamma_k \right)  \nonumber \\
&~ +\frac{1}{2\beta}\sum\nolimits_{k=1}^{K} \left\|\mathbf{HQ}\mathbf{w}_k-\mathbf{v}_k+\beta{\bm\xi_k}\right\|^2  \\
{\rm{s.t.}}&~ \sum\nolimits_{k=1}^{K}\|\mathbf{v}_k\|^2\leq P_\text{max},
\end{align}
\end{subequations}
where $\bm\Omega=\sum_{i=1}^{K} (1+\rho_i) \tilde{\mathbf{G}}_i \tilde{\mathbf{G}}_i^{\mathsf H} \bm\gamma_i^{\mathsf H} \bm\gamma_i$. Note that ${\mathcal P}_{\mathbf{v}_k}$ is a standard convex problem, whose relevant Lagrangian function can be written as
\begin{align}\label{L_v}
\mathcal{L}(\mathbf{v}_k,\lambda)=f_{\mathbf{V}}(\mathbf{v}_k) + \lambda_k \left(\sum_{k=1}^{K}\mathbf{v}_k^{\mathsf H} \mathbf{v}_k-P_\text{max} \right),
\end{align}
where $\lambda_k \geq 0$ is the Lagrangian multiplier. By checking the first-order optimality condition, the optimal $\mathbf{v}_k^\star$ is given by
\begin{align} \label{opt_V}
\mathbf{v}_k^\star=\left(\bm\Psi+\lambda_k \mathbf{I}_L\right)^{-1} \bm\phi_k,
\end{align}
where $\bm\Psi=\bm\Omega+\frac{1}{2\beta}\mathbf{I}$ and $\bm\phi_k=(1+\rho_k) \tilde{\mathbf{G}}_k \bm\gamma_k + \frac{1}{2\beta}(\mathbf{HQ}\mathbf{w}_k+\beta \bm\xi_k)$. If $\mathbf{v}_k$ satisfies the constraint $\sum\nolimits_{k=1}^{K}\|\mathbf{v}_k\|^2\leq P_\text{max}$ when $\lambda_k = 0$, then we have $\lambda_k = 0$. Otherwise, the optimal $\lambda_k$ is determined by checking the complementary slackness $\mathsf{tr}(\mathbf{V}^\mathsf{H} \mathbf{V})-P_\text{max}=0$. Specifically, by performing singular value decomposition (SVD) on matrix $\bm\Psi$, i.e., $\bm\Psi=\mathbf{U}\bm\Lambda \mathbf{U}^{\mathsf H}$, where  $\bm\Lambda=\mathsf{diag}(\Lambda_1,\dots,\Lambda_L)$ with $\{\Lambda_1,\dots,\Lambda_L\}$ being the singular values, we have $\mathbf{V}=\mathbf{U}\left(\bm\Lambda+\lambda_k \mathbf{I}_L\right)^{-1}\mathbf{U}^{\mathsf H} \bm\Phi$ with $\bm\Phi=\left[\bm\phi_1,\dots,\bm\phi_K\right]$. By denoting $\mathbf{X}=\mathbf{U}^{\mathsf H}\bm\Phi \bm\Phi^ {\mathsf H} \mathbf{U}$ and after some basic mathematical
manipulations, the complementary slackness can be rewritten by $\sum_{i=1}^{L} \frac{[\mathbf{X}]_{i,i}}{(\Lambda_i+\lambda_k)^2} -P_\text{max}=0$. It is noted that $\sum_{i=1}^{L} \frac{[\mathbf{X}]_{i,i}}{(\Lambda_i+\lambda_k)^2}$ decreases versus $\lambda_k$, thus the optimal $\lambda_k$ in \eqref{opt_V} can be obtained efficiently by bisection search. Then $\mathbf{v}_k^\star$ follows immediately.

\textbf{Step 4}, we optimize $\mathbf{W}$ with other variables fixed. The subproblem of optimizing $\mathbf{W}$ is given by
\begin{align}
{\mathcal P}_{\mathbf{W}}: ~\min_{\mathbf{W}}~ \left\|\mathbf{HQW}-\mathbf{V}+\beta{\bm\Xi}\right\|_\text{F}^2.
\end{align}
By checking the first-order optimality condition of the objective function, we obtain the optimal $\mathbf{W}$ as follows
\begin{align}\label{opt_W}
\mathbf{W}^{\star}=\left(\mathbf{Q}^{\mathsf H} \mathbf{H}^{\mathsf H} \mathbf{HQ}\right)^{-1} \mathbf{Q}^{\mathsf H} \mathbf{H}^{\mathsf H} (\mathbf{V}-\beta\bm\Xi).
\end{align}

\textbf{Step 5}, we optimize $\mathbf{Q}$ by fixing other variables, which yields the following problem:
\begin{subequations}\label{P_Q}
\begin{align}
{\mathcal P}_{\mathbf{Q}}: ~\min_{\mathbf{Q}} &~
\left\|\mathbf{HQW}-\mathbf{V}+\beta{\bm\Xi}\right\|_\text{F}^2  \\
{\rm{s.t.}}&~ \eqref{cons_Q},\eqref{cons_Q2},
\end{align}
\end{subequations}
By rewriting $\mathbf{HQ}$ as $\tilde{\mathbf{Q}}\tilde{\mathbf{H}}$ and after some basic mathematical manipulations, we arrive at the following problem:
\begin{subequations}\label{Pq_d}
\begin{align}
{\mathcal P}_{\mathbf{q}}^{\text{d}}: ~\max_{\mathbf{q}} &~ 2\Re(\mathbf{q}^{\mathsf H}\mathbf{c}_2)-\mathbf{q}^{\mathsf H} \mathbf{D}_2 \mathbf{q} \label{Pq_d_obj}\\
{\rm{s.t.}}&~ q_n \in \mathcal{Q}, \forall n \in \mathcal{N}, \label{Pq_d_cons}
\end{align}
\end{subequations}
where $\mathbf{D}_2=\mathbf{I}_N \odot (\tilde{\mathbf{H}}\mathbf{W}\mathbf{W}^{\mathsf H} \tilde{\mathbf{H}}^{\mathsf H})$ and $\mathbf{c}_2=\mathsf{diag}\left((\mathbf{V}-\beta\bm\Xi)\mathbf{W}^{\mathsf H}\tilde{\mathbf{H}}^{\mathsf H}\right)$. It is observed that the problem $\mathcal{P}_{\mathbf{q}}^\text{d}$ has the same form as $\mathcal{P}_{\mathbf{q}}^\text{sic}$ in \eqref{Pq_sic}, thus we could use the proposed EWR-based method in Algorithm \ref{Algorithm1} to solve it.

By performing Step 1--5 in sequence until the objective function in \eqref{P_AL_obj} converges, one inner loop is completed.

After examining the inner loop of the PDD-based method, we now turn our attention to the outer loop. In this loop, we begin by calculating the constraint violation $h$ as follows
\begin{align}\label{h}
h=\left\|\mathbf{HQW}-\mathbf{V}\right\|_\text{F}^2.
\end{align}
The value of $h$ determines whether to update the dual variable $\bm\Xi$ or the penalty parameter $\beta$. Specifically, the dual variable is updated by $\bm\Xi^{t+1}=\frac{1}{\beta}(\mathbf{HQW}-\mathbf{V}) +\bm\Xi^{t}$ and the penalty parameter is updated by $\beta^{t+1}=c\beta^t$ \cite{Shi2020}, where $t$ is the outer iteration index and $c$ is the scaling factor. The outer loop terminates until the constraint violation $h$ belows a threshold.

\begin{algorithm}[!t]
    \caption{PDD-based method}
    \label{Algorithm3}
    Initialize primary variables $\left\{\bm\rho, \bm\Gamma, \mathbf{W}, \mathbf{Q}, \mathbf{V}\right\}$, dual variable $\bm\Xi$, constraint violation $h$, threshold $\epsilon$, $\eta$, scaling factor $c_1<1$, $c_2<1$, outer iteration index $t=0$, and penalty factor $\beta>0$\;
    \Repeat{$h< \epsilon$}
    {
    \Repeat{the objective function in \eqref{P_AL_obj} converges}
    {
    Update $\bm\rho$ based on \eqref{opt_rho}\;
    Update $\bm\Gamma$ based on \eqref{opt_Gamma}\;
    Update $\mathbf{V}$ based on \eqref{opt_V}\;
    Update $\mathbf{W}$ based on \eqref{opt_W}\;
    Update $\mathbf{Q}$ based on Algorithm \ref{Algorithm1}\;
    }
    Calculate constraint violation $h$ based on \eqref{h}\;
    \eIf{$h < \eta^t$}{
    Update $\bm\Xi^{t+1}=\frac{1}{\beta}(\mathbf{HQW}-\mathbf{V}) +\bm\Xi^{t}$\;} {
    Update $\beta^{t+1}=c_1\beta^t$\;
    }
    Update $\eta^{t+1}=c_2h$, $t=t+1$\;
    }
    \textbf{Output}: $\mathbf{W}$, $\mathbf{Q}$.
\end{algorithm}

\subsection{Convergence and Complexity Analysis}
The overall proposed PDD-based method is summarized in Algorithm \ref{Algorithm3}, which is guaranteed to converge to the set of stationary solutions of problem $\mathcal{P}_\text{d}$ \cite{Shi2020}. In terms of computational complexity, it is mainly determined by the complexity of updating variables in the inner loop. For each iteration of the inner loop, the computational complexity for optimizing each variable in $\left\{ \bm\rho,\bm\Gamma,\mathbf{V},\mathbf{W},\mathbf{Q}\right\}$ can be estimated sequentially as $\mathcal{O}(K^2N^2)$, $\mathcal{O}(K^2N^2)$, $\mathcal{O}(KN^3)$, $\mathcal{O}(L^2 N)$, $\mathcal{O}(I_\text{EWR}^\text{d}N^2)$, respectively, where $I_\text{EWR}^\text{d}$ is the iteration number of the EWR-based method for solving problem $\mathcal{P}_\mathbf{q}^\text{d}$ in \eqref{Pq_d}. In summary, the overall complexity of Algorithm \ref{Algorithm3} is evaluated by $\mathcal{O}\left(I_\text{in} I_\text{out} \left(K^2N^2+K^2N^2+KN^3 +L^2 N+ I_\text{EWR}^\text{d}N^2\right)\right)$, which is in polynomial time.

\section{Numerical Results}   \label{Section5}
{\color{black}In this section, we consider a narrowband system and provide simulation results to validate the performance of the proposed algorithms and reveal some important insights. The specific simulation parameters are listed in TABLE \ref{table3}.} For the considered Rician fading channel, we model the large scale path-loss coefficient $\alpha_k$ as $\alpha_k=\alpha_0\left(\frac{D_k}{D_0}\right)^{-\Gamma_k}$, where $\alpha_0=-30\text{dB}$ is the path loss at reference distance $D_0=1\text{m}$, $\Gamma_k=2.5$ is the path loss exponents of the propagation environment, and $D_k$ is the distance from the UE $k$ to the BS. {\color{black}We place DMA on the (x,z) plane at the location of $(0, 0, 20\text{m})$ to cover more users. Then the LOS component $\bar{\mathbf{g}}_k$ can be expressed by the array response of the uniform planar array (UPA), which is given by
\begin{align}\label{LOS}
\bar{\mathbf{g}}_k&=\left[1,\dots, {\rm{e}}^{{\rm{j}}\frac{2\pi}{\lambda_c}\big(\sin \omega \cos \psi {\rm{mod}}(n-1,S) d_x +\cos\omega\lfloor n-1,S\rfloor d_z\big)} ,\right. \nonumber \\
& \left.\dots,{\rm{e}}^{{\rm{j}}\frac{2\pi} {\lambda_c} \big(\sin\omega \cos\psi (S-1)d_x + \cos\omega (L-1)d_z \big)}\right]^{\mathsf T},
\end{align}
where $\lambda_c$ is the wavelength of the carrier wave, $d_z$ is the spacing between each microstrip, $d_x$ is the spacing between each element on every microstrip, $\psi$ denotes the azimuth angles of arrival (AoA), and $\omega$ denotes elevation AoA, respectively. The users are uniformly distributed on a circle centered at $(0, 200\text{m}, 0)$ with radius $d_0 =20\text{m}$, as shown in {\figurename} {\ref{Simulation_setup}}.} Besides, the well-known Kronecker model is considered for $\mathbf{R}_k$, i.e., $\mathbf{R}_k= \mathbf{R}_{\text{H},k}\otimes\mathbf{R}_{\text{V},k}$. Here, $\mathbf{R}_{\text{H},k}$ and $\mathbf{R}_{\text{V},k}$ are the spatial correlation matrices for user $k$ of the horizontal and vertical domains, respectively, and are defined by the  exponential correlation model as follows:
\begin{align}\label{correlation}
[\mathbf{R}]_{i,j}=\left\{\begin{array}{ll}
r^{i-j}, & i\leq j \\
\left[\mathbf{R}\right]_{j,i} , & i>j
\end{array},\right.
\end{align}
where $0<r<1$ is the correlation coefficient. For simplicity, we set the correlation coefficient $r=0.7$ for $\forall k \in \mathcal{K}$. Unless otherwise specified, we set $K=4$ and $K_0=10$ dB to simulate suburban and open road scenario \cite{Han2019}. All the results are averaged over $10^4$ independent channel realizations.

\begin{table}[!h]
\caption{Simulation parameters.}
\label{table3}
\centering
{\color{black}
\begin{tabular}{|>{\raggedright\arraybackslash}p{3.8cm}|> {\raggedright\arraybackslash}p{3.3cm}|}
\hline
\textbf{Parameters} & \textbf{Values} \\ \hline
Rician factor &  $K_0=10$ \cite{Han2019} \\ \hline
Reference distance & $D_0=1$m \\ \hline
Reference path loss & $\alpha_0=-30$dB \\ \hline
Path loss exponents & $\Gamma_k=2.5$ \\ \hline
DMA position &  $(0,0,20\text{m})$ \\ \hline
Wavelength & $\lambda_c=1.07$cm \cite{Zhang2022} \\ \hline
Spacing between each element & $d_x=\lambda_c/2$ \\ \hline
Spacing between each microstrip & $d_z=\lambda_c/2$ \\ \hline
User distribution center & $(0,200\text{m},0)$ \\ \hline
User distribution radius & $d_0=20$m \\ \hline
DMA attenuation coefficient & $\alpha=0.6\text{m}^{-1}$ \cite{Kimaryo2023} \\ \hline
Wavenumber in microstrip & $\beta=827.67\text{m}^{-1}$ \cite{Kimaryo2023} \\ \hline
Correlation coefficient & $r=0.7$ \\ \hline
Noise power at the BS & $N_0=-80$dB \\ \hline
Noise power at the UE & $N_k=-80$dBm \\ \hline
Power consumption of RF chains  &$P_\text{RF}=27$ dBm \\ \hline
Power consumption of BS  & $P_\text{BS}=39$ dBm \\ \hline
Power consumption of the phase shifters & $P_\text{PS}=17$ dBm \\ \hline
Amplifier efficiency factor & $\epsilon=0.35$ \\ \hline
Number of UEs & $K=4$ \\ \hline
Penalty factor & $\beta=10^5$ \\ \hline
PDD scaling factors & $c_1=0.5$, $c_2=1/6$ \\ \hline
PDD threshold & $\epsilon=10^{-5}$ \\ \hline
\end{tabular}}
\end{table}

\begin{figure}[!t]
\centering
\setlength{\abovecaptionskip}{0pt}
\includegraphics[height=0.25\textwidth]{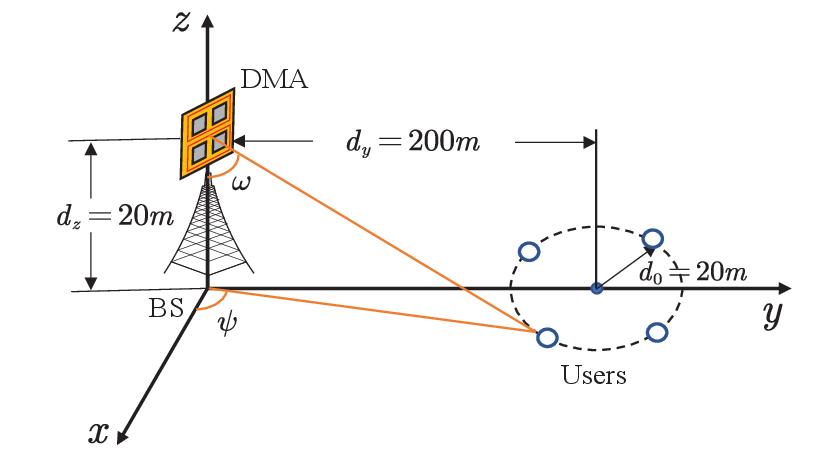}
\caption{Simulation setup.}
\label{Simulation_setup}
\end{figure}

\subsection{Uplink Transmission}
We first evaluate the performance and convergence of the Algorithm \ref{Algorithm2} for uplink transmission with statistical CSI.

\begin{figure}[!t]
\centering
\setlength{\abovecaptionskip}{0pt}
\includegraphics[height=0.32\textwidth]{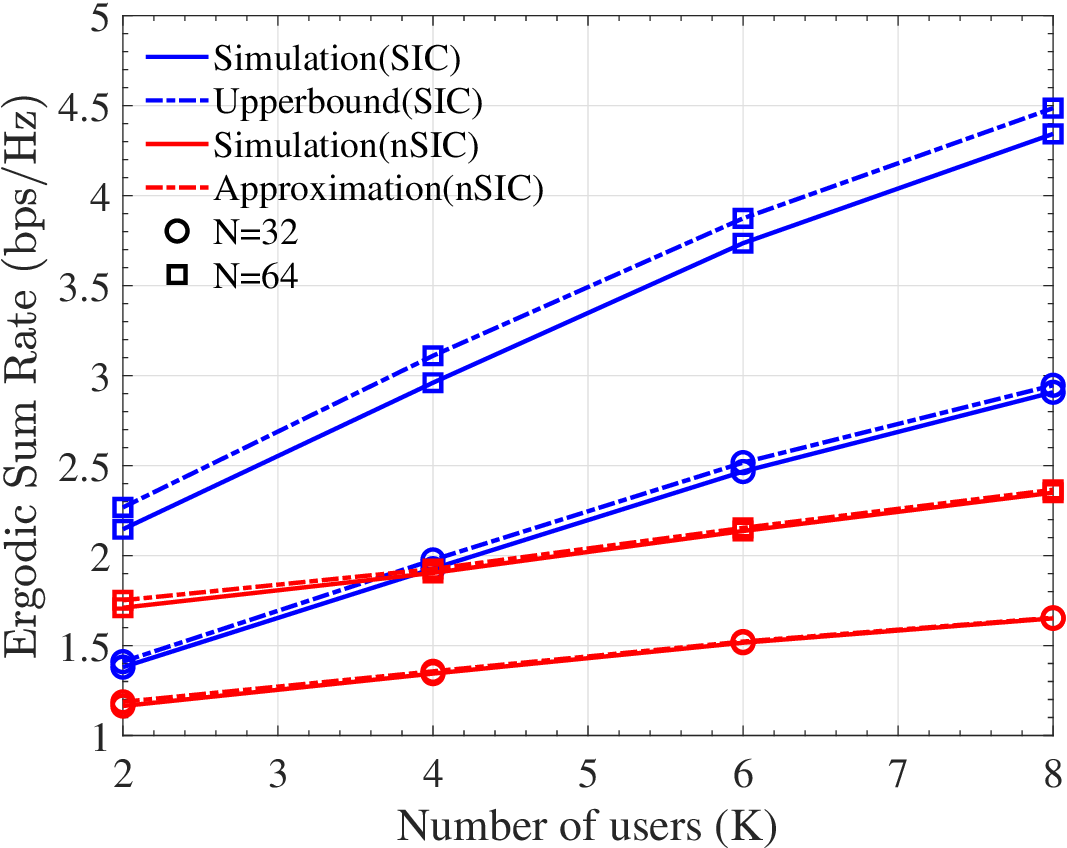}
\caption{The validation of the uplink ergodic sum rate.}
\label{Analytic_up}
\end{figure}

\subsubsection{Validation of the Asymptotic Analysis}
We first use {\figurename} {\ref{Analytic_up}} to validate the conclusions in Theorem \ref{Theorem_1} and Theorem \ref{Theorem_2}. As shown in {\figurename} {\ref{Analytic_up}}, the simulated results with MMSE-SIC decoding are tightly upper bounded by the analytical results, which proves the correctness of the Theorem \ref{Theorem_1}. It is noted that increasing the number of DMA elements can enhance the ergodic sum rate but reduce the tightness of the upper bound. For the case of MMSE-nSIC decoding, the approximated and simulated results almost overlap, which validates the accuracy of the analytical approximation in Theorem \ref{Theorem_2}.

\begin{figure}[!t]
\centering
\setlength{\abovecaptionskip}{0pt}
\includegraphics[height=0.32\textwidth]{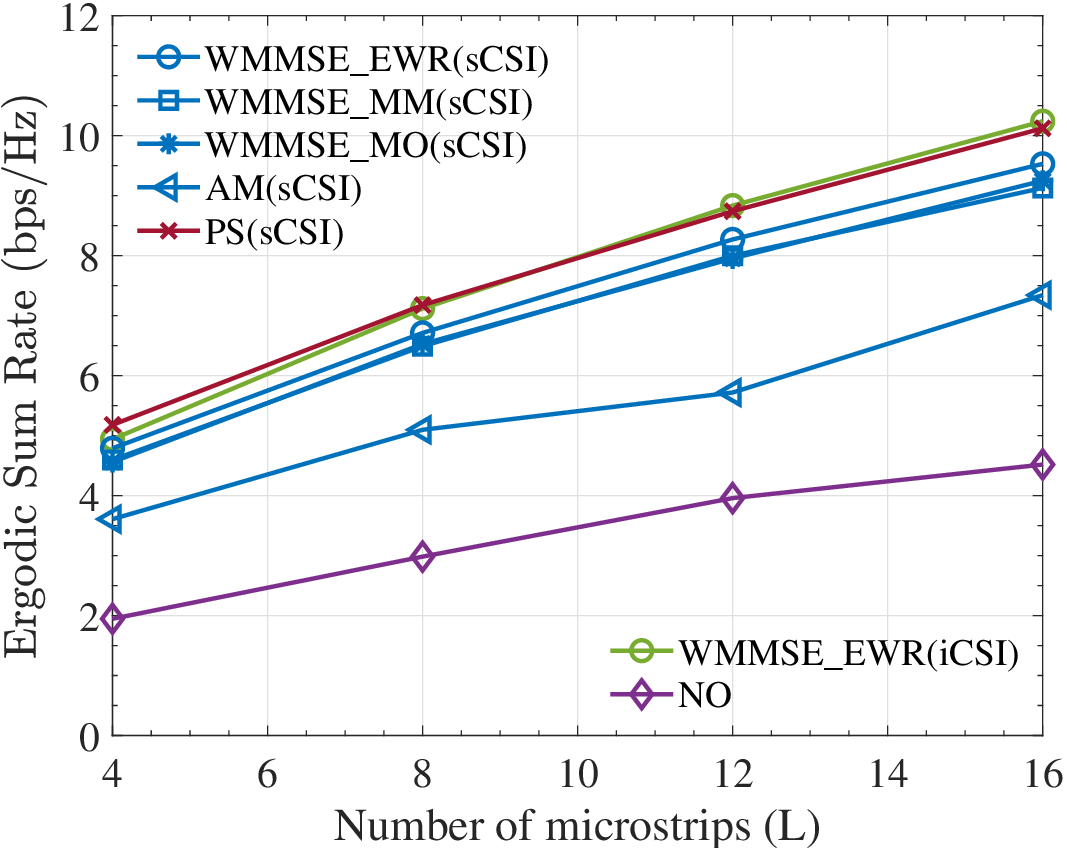}
\caption{The performance comparison with benchmarks.}
\label{Compare_up}
\end{figure}

\subsubsection{Performance Evaluation}
In {\figurename} {\ref{Compare_up}}, we aim to compare our proposed WMMSE-based method with some latest baselines. Due to the fact that most articles consider the case of SIC, we only evaluate the performance of SIC here, i.e., $\mathcal{R}_\text{sic}$ in \eqref{R_SIC}. Specifically, we assume each microstrip is definitely composed of $S=8$ and change the number of DMA elements $N$ by setting different microstrip numbers $L$. {\color{black}The curves of the following schemes are presented:
\begin{itemize}
\item WMMSE\_EWR (sCSI): The proposed WMMSE-based method exploiting only statistical CSI, in which the DMA coefficient matrix is optimized using EWR-based approach described in Algorithm \ref{Algorithm1}.
\item WMMSE\_MM (sCSI): A WMMSE-based method with statistical CSI, where the DMA matrix is optimized via the minorization-maximization (MM) method \cite{Xu2024}.
\item WMMSE\_MO (sCSI): A WMMSE-based method with statistical CSI, where the DMA matrix is optimized using the manifold optimization (MO)-based method \cite{Kimaryo2023,Chen2025}.
\item AM (sCSI): The alternating minimization (AM) algorithm based on statistical CSI \cite{Shlezinger2019,You2023}.
\item PS (sCSI): The proposed WMMSE-based method with statistical CSI, where the DMA is replaced by a partially-connected phase shifter (PS)-based multi-antenna array \cite{Zhang2022}.
\item WMMSE\_EWR (iCSI): The WMMSE-based method with instantaneous CSI \cite{Shi2011}, where the DMA coefficient matrix is optimized by the proposed EWR-based approach in Algorithm \ref{Algorithm1}.
\item NO: A baseline scheme in which the $q_{l,s}$ in \eqref{cons_Q} is set as $1$ for $\forall l \in \mathcal{L}$ and $\forall s \in \mathcal{S}$.
\end{itemize}

It is noted that the proposed WMMSE-based algorithm effectively improves the ergodic sum rate and outperforms the baseline AM-based method, which demonstrates the superiority of our proposed algorithm. This outcome is expected, as the relaxation employed in the AM algorithm is relatively inaccurate. Besides, the proposed EWR-based method for optimizing DMA weights matrix performs better than the MM- and MO-based approaches, which further validates the effectiveness of our EWR-based algorithm. Compared to the instantaneous CSI scenario, the performance loss when using only statistical CSI is acceptable, as the simulation assumes a dominant LoS component. Therefore, in scenarios with strong LoS component, relying on statistical CSI offers a favorable trade-off between performance and complexity. It is also noted that the DMA-enabled uplink system underperforms relative to the PS-based multi-antenna system due to its limited phase shift adjustment capability. However, owing to the substantially lower energy consumption of DMA elements compared to phase shifters, DMA-enabled systems achieve superior energy efficiency \cite{You2023,Chen2025}.}

\begin{figure}[!t]
\centering
\setlength{\abovecaptionskip}{5pt}
\includegraphics[height=0.32\textwidth]{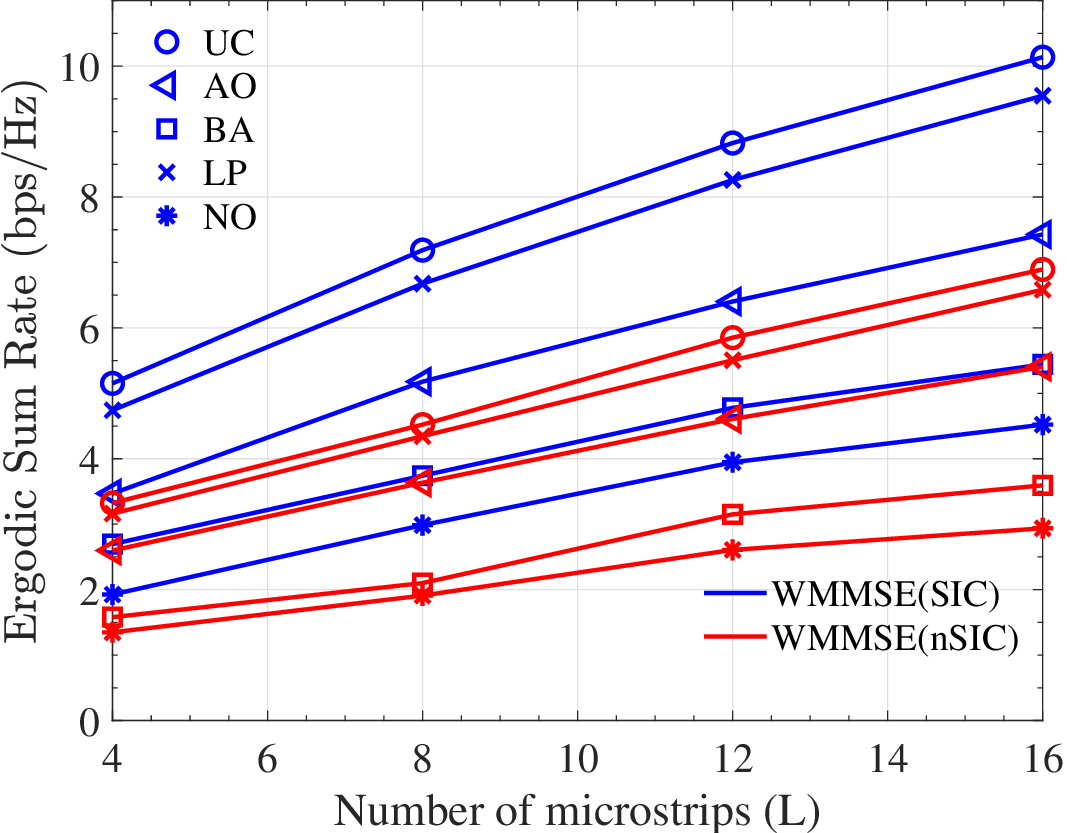}
\caption{The impact of different sets of $\mathcal{Q}$.}
\label{Sets_Q_up}
\end{figure}

\subsubsection{Impact of Different Sets of $\mathcal{Q}$}
{\figurename} {\ref{Sets_Q_up}} illustrates the sum rate performance achieved by different sets $\mathcal{Q}$:
\begin{itemize}
  \item UC: Unconstrained weights, i.e., $\mathcal{Q}=\mathbbmss{C}$;
  \item AO: Amplitude only, here $\mathcal{Q}=[0.001,5]$ refer to \cite{Shlezinger2019};
  \item BA: Binary amplitude, here $\mathcal{Q}=\{0,0.1\}$ refer to \cite{Shlezinger2019};
  \item LP: The considered Lorentzian phase constraint.
\end{itemize}
For the unconstrained weights, the weight matrix $\mathbf{Q}$ can be obtained by exploiting the first-order optimal condition. Besides, it is noted that the objective function in \eqref{Pn_obj} is a quadratic function respecting the real variable $q_n$, which is easy to optimize. Thus, for the case of AO and BA constraints, the problem ${\mathcal P}_{\mathbf{q}}^{\text{sic}}$ and ${\mathcal P}_{\mathbf{q}}^{\text{nsic}}$ can also be solved by the proposed EWR-based method, which highlights the universality and superiority of our proposed algorithm. {\color{black}Moreover, it is observed that SIC decoding significantly outperforms nSIC decoding. Thus, assuming uniform user transmission power, SIC decoding is the preferred choice, offering better performance at the cost of higher computational complexity.}

\begin{figure}[!t]
\centering
\setlength{\abovecaptionskip}{5pt}
\includegraphics[height=0.32\textwidth]{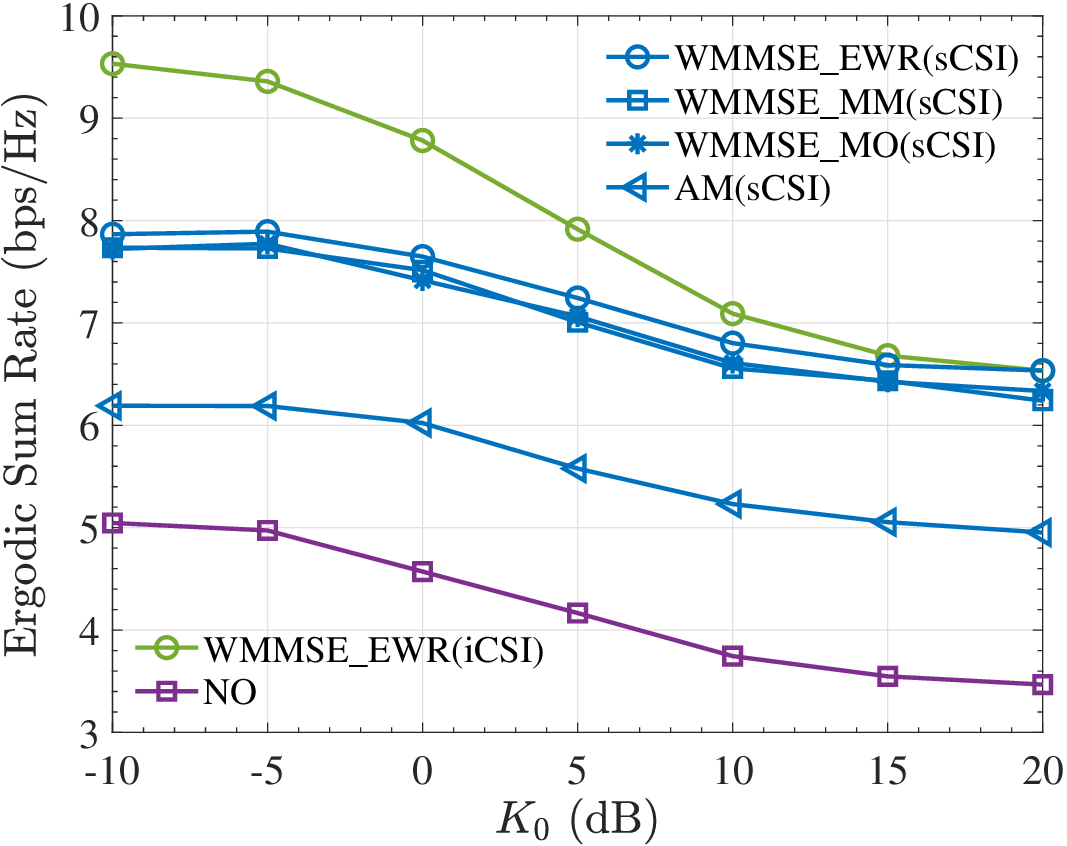}
\caption{The ergodic sum rate versus Rician factors with $N=64$ for uplink transmission.}
\label{Uplink_versus_K}
\end{figure}
{\color{black}
\subsubsection{Impact of the Rician Factor $K_0$}
In {\figurename} {\ref{Uplink_versus_K}}, the uplink ergodic sum rate is plotted versus the Rician factor $K_0$ with $N=64$. As shown, the ergodic sum rate decreases with increasing $K_0$. This trend arises because the users are located relatively close to one another, resulting in similar LoS components between different users and the BS, which leads to more severe inter-user interference as $K_0$ increases. It is also observed that the proposed WMMSE-based and EWR-based methods perform better than the benchmarks schemes. Moreover, the performance gap between schemes utilizing statistical CSI and those based on instantaneous CSI widens as $K_0$ decreases, due to the increased randomness in the channel. Therefore, we conclude that statistical CSI-based designs become more favorable as $K_0$ increases.}

\begin{figure}[!t]
\centering
\setlength{\abovecaptionskip}{5pt}
\includegraphics[height=0.32\textwidth]{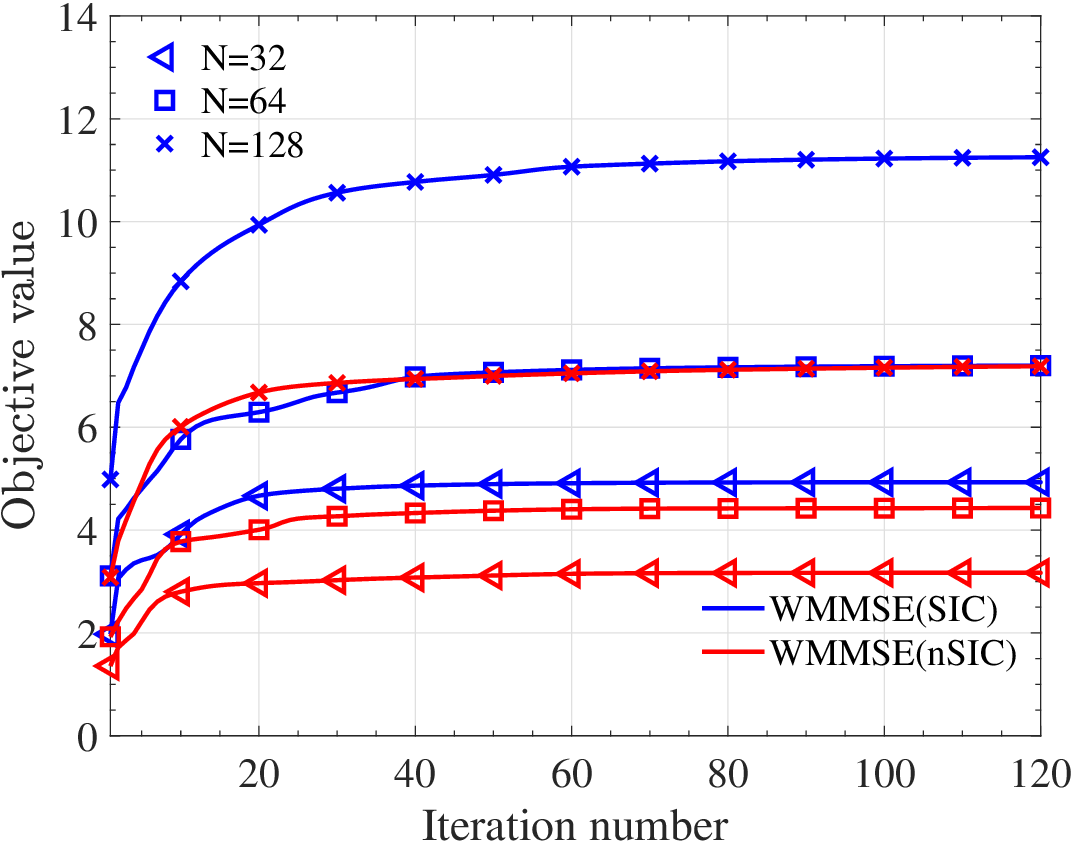}
\caption{Convergence behavior for uplink transmission.}
\label{Convergence_up}
\vspace{-10pt}
\end{figure}

\subsubsection{Convergence Evaluation}
{\figurename} {\ref{Convergence_up}} demonstrates the convergence performance of the proposed WMMSE-based method in Algorithm \ref{Algorithm2} with different numbers of DMA elements $N$. It is shown that the increasing $N$ could reach larger SE performance. Besides, the ergodic sum rate achieved by the MMSE-SIC decoder is higher than that achieved by the MMSE-nSIC decoder, but the convergence speed of both two cases is similar.

\subsection{Downlink Transmission}
Having evaluated the performance for the uplink transmission, we then appraise the SE performance of Algorithm \ref{Algorithm3} for the downlink transmission with statistical CSI. For the PDD-based algorithm, we set the initial penalty factor $\beta=10^5$, initial $h=1$, $\eta=1$, scaling factors $c_1=0.5,c_2=1/6$, and threshold $\epsilon=10^{-5}$. Unless otherwise specified, we assume $S=8$ and $L=8$.

\begin{figure}[!t]
\centering
\setlength{\abovecaptionskip}{5pt}
\includegraphics[height=0.32\textwidth]{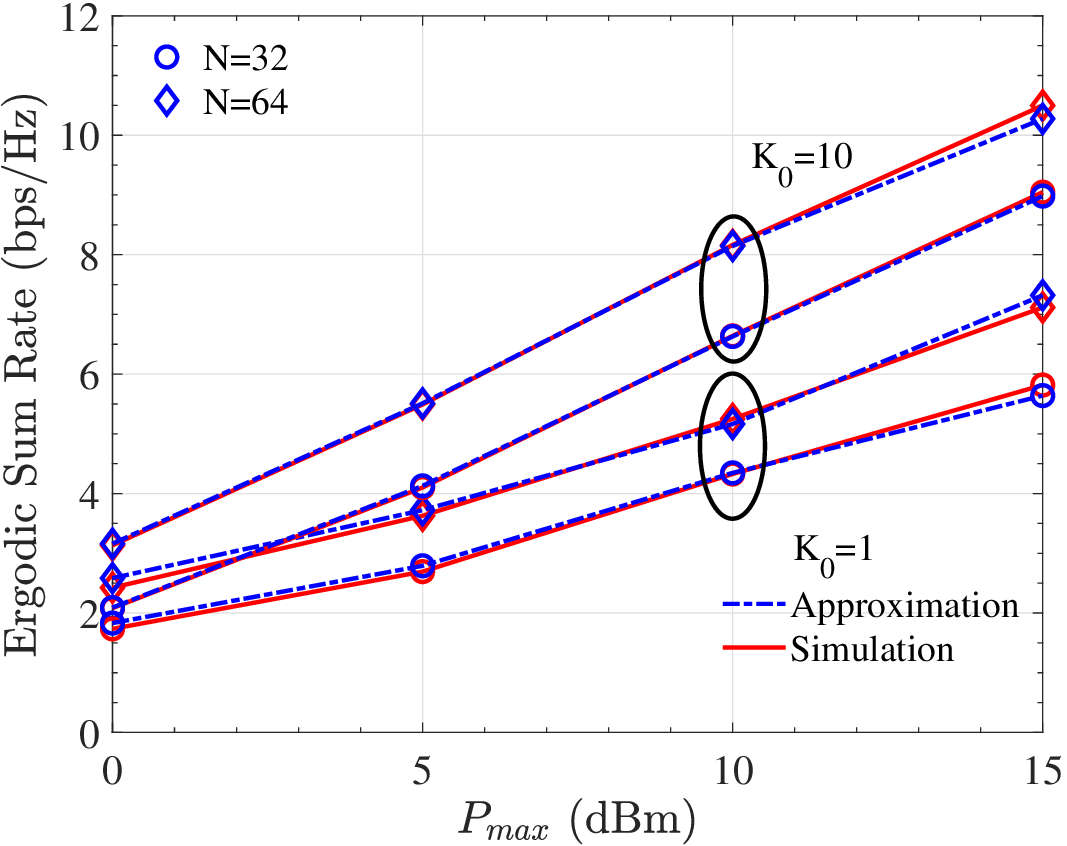}
\caption{The validation of the downlink ergodic sum rate.}
\label{Analytic_d}
\end{figure}

\subsubsection{Validation of the Asymptotic Analysis}
We first use {\figurename} {\ref{Analytic_d}} to evaluate the accuracy of the asymptotic analysis comprehensively. {\color{black}Specifically, we use the proposed PDD-based algorithm to optimize the phase shifts of the DMA and digital beamforming while considering different Rician factor $K_0$ and the number of DMA elements $N$. As illustrated, the analytical and simulated curves match well in the cases of different number of elements and Rician factors, which validates the accuracy of the approximation in Corollary \ref{corollary1}.}

\begin{figure}[!t]
\centering
\setlength{\abovecaptionskip}{5pt}
\includegraphics[height=0.32\textwidth]{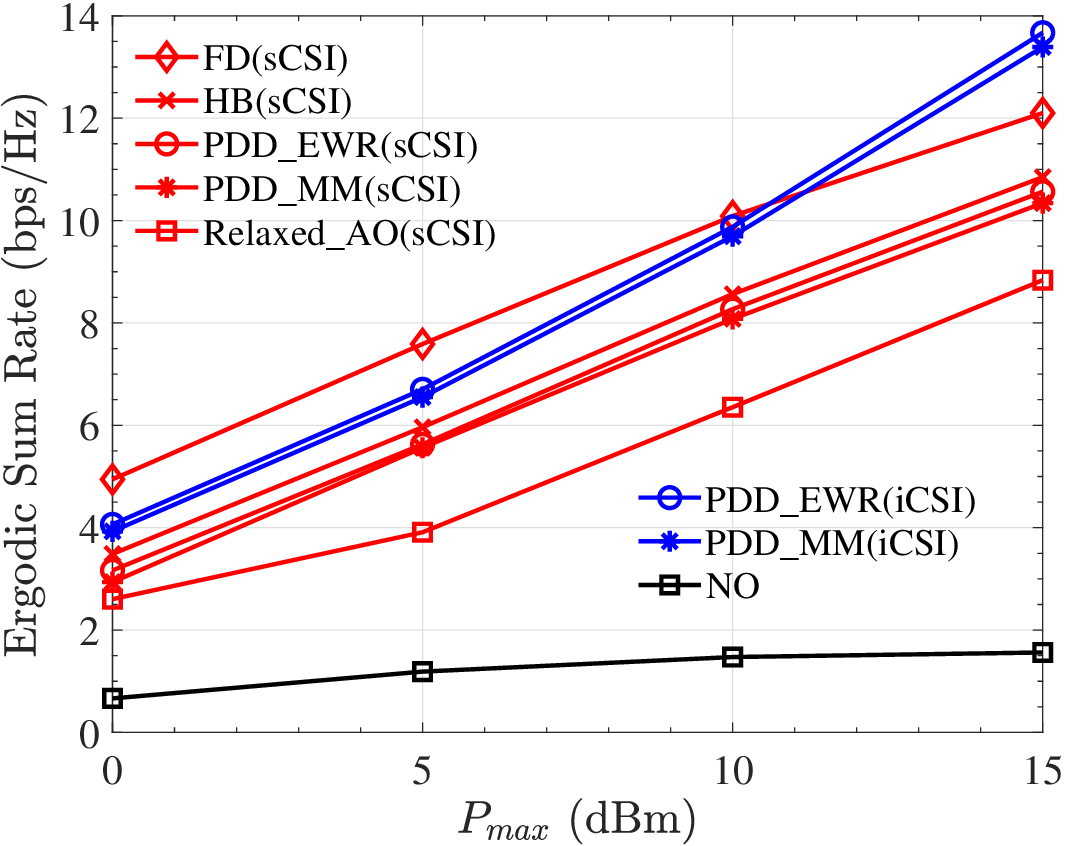}
\caption{The performance comparison with benchmarks.}
\label{Compare_d}
\end{figure}
{\color{black}
\subsubsection{Performance Evaluation}
In {\figurename} {\ref{Compare_d}}, we evaluate the ergodic sum rate $\mathcal{R}_\text{d}$ in \eqref{R_d}, achieved by the proposed PDD-based algorithm. The following schemes are included  for comparison:
\begin{itemize}
\item FD (sCSI): Full-digital beamforming based on statistical CSI \cite{Zhang2022}.
\item HB (sCSI): Hybrid beamforming with partially-connected phase shifters, using statistical CSI \cite{Zhang2022}.
\item PDD\_EWR (sCSI): The proposed PDD-based algorithm with statistical CSI, where the DMA coefficient matrix is optimized using the proposed EWR-based method.
\item PDD\_MM (sCSI): The PDD-based algorithm with statistical CSI, where the DMA coefficient matrix is optimized via the MM method \cite{Xu2024}.
\item Relaxed\_AO (sCSI): An alternating optimization algorithm with statistical CSI, in which the digital precoder $\mathbf{W}$ and DMA weight matrix $\mathbf{Q}$ are optimized sequentially by relaxing the power constraint to $\sum\nolimits_{k=1}^{K}\|\mathbf{w}_k\|^2\leq P_\text{max}$ \cite{Zhang2021,Zhang2022,Kimaryo2023}.
\item PDD\_EWR (iCSI): The PDD-based algorithm with instantaneous CSI \cite{Shi2020}, where the DMA coefficient matrix is optimized using the proposed EWR-based method.
\item PDD\_MM (iCSI): The PDD-based algorithm with instantaneous CSI \cite{Shi2020}, where the DMA coefficient matrix is optimized via the MM method.
\item NO: A baseline scheme is considered, where the DMA coefficient $q_{l,s}$ in \eqref{cons_Q} is set as $1$ for $\forall l \in \mathcal{L}$ and $\forall s \in \mathcal{S}$, and the digital precoding matrix $\mathbf{W}$ adopts uniform power allocation.
\end{itemize}

As shown, the proposed PDD-based algorithm outperforms the relaxed alternating optimization method. Furthermore, within the same PDD framework, the proposed EWR-based algorithm performs achieves slightly better performance than the MM-based approach. These results demonstrate the effectiveness of the proposed algorithms. It is noted that the performance gap between the schemes using instantaneous CSI and statistical CSI increases with transmit power, but remains within an acceptable range. In addition, the DMA-enabled downlink system performs slightly worse than the PS-based hybrid A/D beamforming system due to its limited beamforming flexibility. However, thanks to the significantly lower power consumption of DMA elements compared to phase shifters, the energy efficiency of DMA-enabled systems surpasses that of PS-based multi-antenna systems \cite{You2023,Chen2025}, as will be analyzed in the next section.
}
\begin{figure}[!t]
\centering
\setlength{\abovecaptionskip}{5pt}
\includegraphics[height=0.32\textwidth]{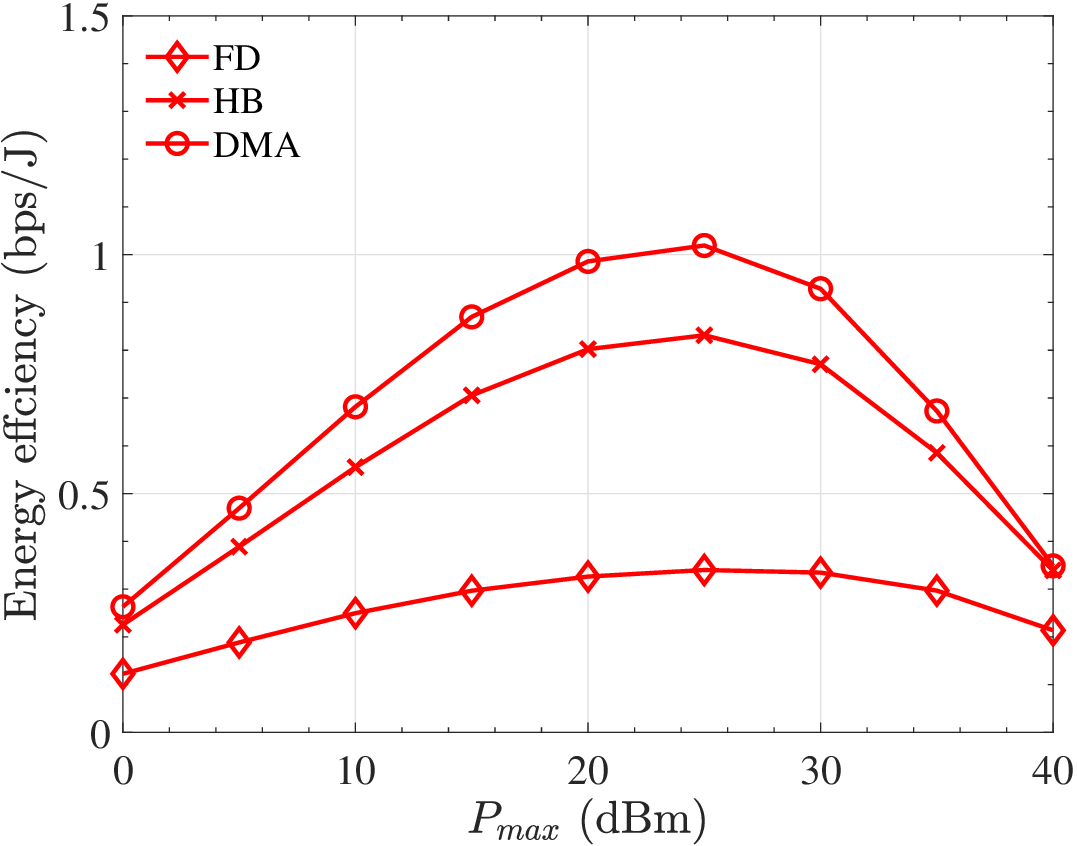}
\caption{Energy efficiency versus $P_\text{max}$.}
\label{EE}
\end{figure}

{\color{black}
\subsubsection{Comparison of EE with Different Antenna Architecture}
In {\figurename} {\ref{EE}}, we compare the EE performance of full-digital antenna systems, PS-based partially-connected hybrid systems, and DMA-enabled systems. Specifically, the EE is defined as
\begin{align}
\eta=\frac{\mathcal{R}_\text{d}}{P_\text{tot}},
\end{align}
where $\mathcal{R}_\text{d}$ is the ergodic sum rate obtained using statistical CSI, and $P_\text{tot}$ denotes the total power consumption. This includes the transmit power $P_\text{max}$, the power consumed by RF chains $P_\text{RF}$ (e.g., amplifiers and mixers), and the static circuit power consumption $P_\text{BS}$ \cite{You2023,Chen2025}. Based on this, the total power consumption of the three systems is given by
\begin{align}
P_\text{tot}^\text{FD}&=\frac{1}{\epsilon}P_\text{max}+NP_\text{RF}+P_\text{BS},\\
P_\text{tot}^\text{HB}&=\frac{1}{\epsilon}P_\text{max}+LP_\text{RF}+NP_\text{PS}+P_\text{BS},\\
P_\text{tot}^\text{DMA}&=\frac{1}{\epsilon}P_\text{max}+LP_\text{RF}+P_\text{BS},
\end{align}
where $\epsilon$ is amplifier efficiency factor and $P_\text{PS}$ represents the power consumption of the phase shifters used in hybrid beamforming. Following the parameters in \cite{Chen2025}, we set $P_\text{RF}=27$ dBm, $P_\text{BS}=39$ dBm, $P_\text{PS}=17$ dBm, and $\epsilon=0.35$. As illustrated, the DMA-enabled systems outperform both PS-based and full-digital antenna systems in terms of EE performance, underscores the advantage of DMA in reducing power consumption, highlighting the potential of DMA as a power-efficient alternative for next-generation wireless systems.
}

\begin{figure}[!t]
\centering
\setlength{\abovecaptionskip}{5pt}
\includegraphics[height=0.32\textwidth]{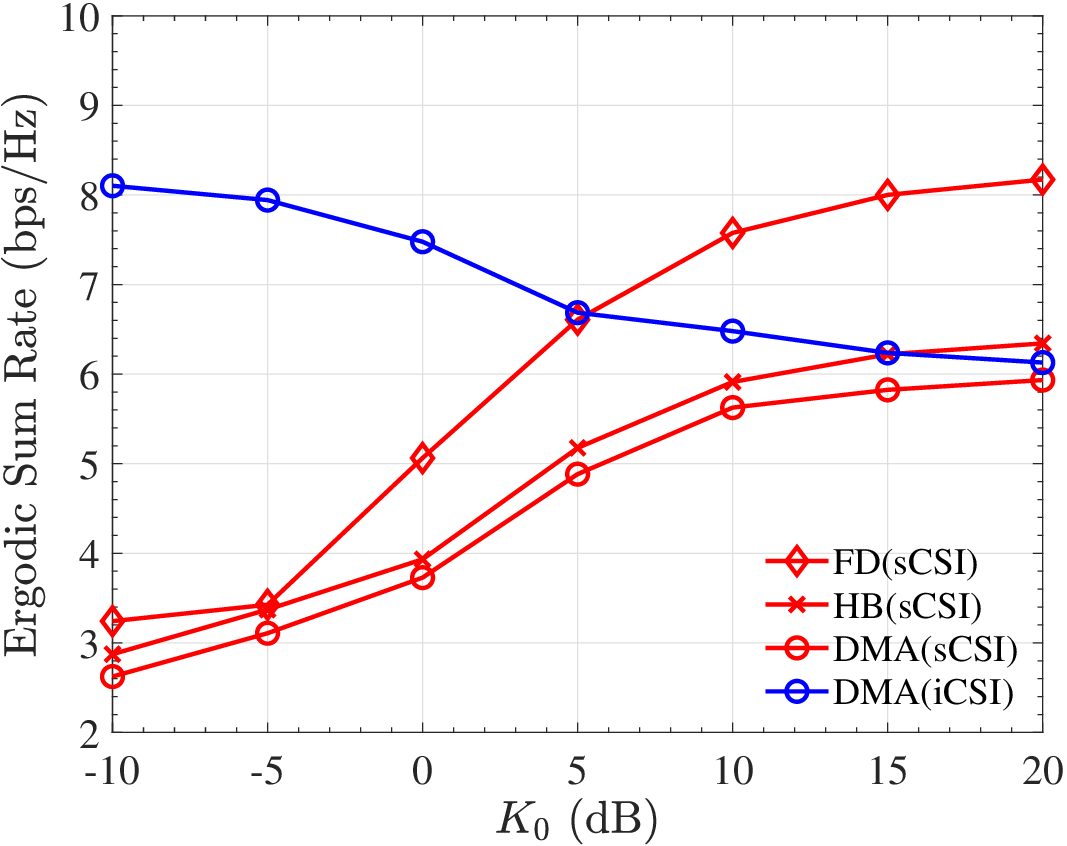}
\caption{Ergodic sum rate versus Rician factors for downlink transmission with $N=64$ and $P_\text{max}=5$ dBm.}
\label{Downlink_versus_K}
\end{figure}
{\color{black}
\subsubsection{Impact of the Rician Factor $K_0$}
To evaluate the impact of the Rician factor $K_0$, we simulate the ergodic sum rate versus $K_0$ in {\figurename} {\ref{Downlink_versus_K}}. It is worth noting that the ergodic sum rate with statistical CSI increases with $K_0$. As $K_0$ grows, the line-of-sight (LoS) component becomes more dominant. Although inter-user interference also increases, the enhanced spatial focusing of the statistical channels leads to improved performance. Moreover, the performance gap between using instantaneous CSI and statistical CSI narrows as $K_0$ increases and tends to vanish when the LoS component is strong. This indicates that exploiting statistical CSI is a practical choice in high-$K_0$ scenarios, achieving a favorable trade-off between performance and channel estimation overhead.

However, when the NLoS component dominates, or in pure Rayleigh fading environments, the performance degradation under statistical CSI becomes significant. In such cases, statistical CSI-based designs are not recommended\footnote{\color{black}When the channel follows Rayleigh fading or Rician fading with a dominant NLoS component, a two-timescale design approach \cite{Zhao2021} can be considered to reduce channel estimation overhead and real-time optimization complexity while maintaining good sum rate performance. This will be investigated in our future work.}.

}

\begin{figure}[!t]
\centering
\setlength{\abovecaptionskip}{5pt}
\includegraphics[width=0.43\textwidth]{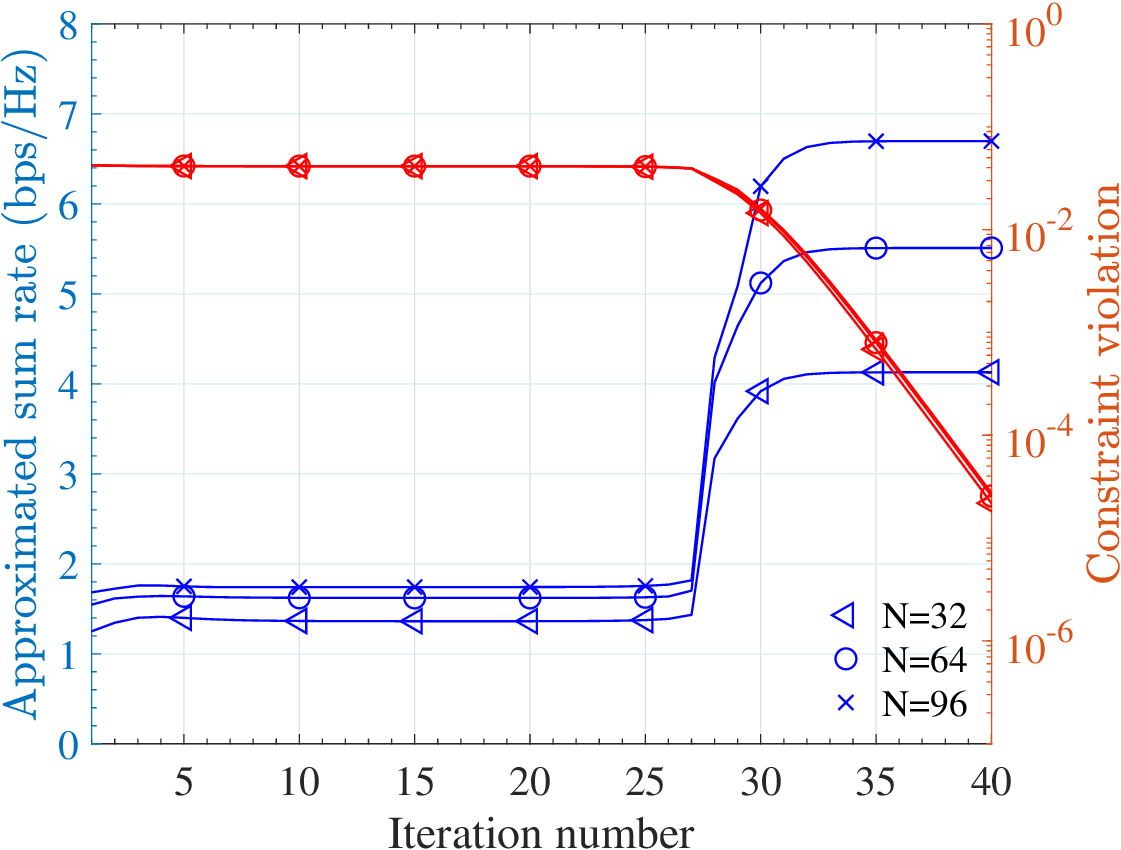}
\caption{Convergence behavior with $P_\text{max}=5$ dBm.}
\label{PDD_convergence}
\end{figure}

\subsubsection{Convergence Performance}
Finally, we use {\figurename} {\ref{PDD_convergence}} to indicate the convergence performance of the proposed PDD-based algorithm for downlink transmission. As shown in {\figurename} {\ref{PDD_convergence}}, the approximated sum rate $\tilde{\mathcal{R}}_\text{d}$ obtained by the PDD-based algorithm converges rapidly in less than $35$ outer iterations. Besides, the constraint violation $h$ reduces to a threshold $\eta=10^{-5}$ in less than $40$ outer iterations, which demonstrates a fast convergence speed.

\section{Conclusion}  \label{Section6}
In this paper, we studied the SE performance of the DMA-enabled MU-MISO wireless systems with statistical CSI, considering both uplink and downlink transmissions. For each scenario, we first formulate effective approximations and then propose efficient optimization algorithms to maximize the derived surrogates, specifically, the WMMSE-based algorithm for the uplink and PDD-based algorithm for the downlink. Numerical results demonstrate that our approximations closely match the simulation outcomes, and the proposed algorithms outperform state-of-the-art baseline schemes in both uplink and downlink settings. {\color{black}For scenarios with a strong LoS component, the proposed statistical CSI-based design achieves satisfactory SE performance, offering a favorable trade-off between performance and channel estimation overhead. However, when the NLoS component dominates, the statistical CSI-based approach suffers significant performance degradation compared to real-time adjustment with instantaneous CSI, particularly in downlink transmissions. Therefore, relying only on statistical CSI for system design in such environments is not recommended.}

\begin{appendix}

\subsection{Proof of Theorem \ref{Theorem_1}}\label{Proof_A}
Since $\log\det(\cdot)$ is a concave function, we could utilize Jensen's inequality to obtain the upper bound of $\mathcal{R}_{\text{sic}}$ in \eqref{R_SIC}:
\begin{align}\label{Jensen}
\mathcal{R}_{\text{sic}} &\leq \overline{\mathcal{R}}_{\text{sic}} \\
&=\log\det\left( \mathbf{I}_L+\mathbf{Q}^{\mathsf H} \mathbf{H}^{\mathsf H} \sum_{k=1}^{K} \mathbbmss{E}\left\{\mathbf{g}_k \mathbf{g}_k^{\mathsf H} \right\} \mathbf{H}\mathbf{Q}(\mathbf{P})^{-1} \right),\nonumber
\end{align}
where $\mathbbmss{E}\left\{\mathbf{g}_k \mathbf{g}_k^{\mathsf H} \right\}$ can be calculated as
\begin{align}\label{expectation}
\mathbbmss{E}\left\{\mathbf{g}_k \mathbf{g}_k^{\mathsf H} \right\}&=\mathbbmss{E}\left\{\left(\sqrt{\alpha_k\frac{K_0}{1+K_0}} \bar{\mathbf{g}}_k+ \sqrt{\alpha_k\frac{1}{1+K_0}}\tilde{\mathbf{g}}_k\right) \right. \nonumber\\
& \left. \cdot \left( \sqrt{\alpha_k\frac{K_0}{1+K_0}}\bar{\mathbf{g}}_k+ \sqrt{\alpha_k\frac{1}{1+K_0}}\tilde{\mathbf{g}}_k\right)^{\mathsf H} \right\} \nonumber \\
&= \frac{\alpha_k K_0}{1+K_0}\bar{\mathbf{g}}_k\bar{\mathbf{g}}_k^{\mathsf H}+\frac{\alpha_k}{1+K_0} \mathbf{R}_k.
\end{align}
By denoting $\tilde{\mathbf{G}}_k=\left[\sqrt{\frac{\alpha_k K_0}{1+K_0}}\bar{\mathbf{g}}_k,\sqrt{\frac{\alpha_k }{1+K_0}}\mathbf{R}_k^{1/2}\right]\in \mathbbmss{C}^{N \times (N+1)}$, we have $\tilde{\mathbf{G}}_k\tilde{\mathbf{G}}_k^{\mathsf H}=\mathbbmss{E}\left\{\mathbf{g}_k \mathbf{g}_k^{\mathsf H} \right\}$. Let $\mathbf{G}=\left[\tilde{\mathbf{G}}_1,\dots,\tilde{\mathbf{G}}_K \right] \in \mathbbmss{C}^{N \times K(N+1)}$, then we can obtain $\mathbf{G}\mathbf{G}^{\mathsf H}=\sum_{k=1}^{K} \tilde{\mathbf{G}}_k\tilde{\mathbf{G}}_k^{\mathsf H}$. The final results follow immediately.

\subsection{Proof of Theorem \ref{Theorem_2}}\label{Proof_B}
By denoting $\mathbf{A}=\sum_{i=1}^{K}\mathbf{Q}^{\mathsf H} \mathbf{H}^{\mathsf H}\mathbf{g}_i \mathbf{g}_i ^{\mathsf H} \mathbf{HQ}+ \mathbf{P}$ and $\mathbf{B}_k=\sum_{i\neq k}^{K}\mathbf{Q}^{\mathsf H} \mathbf{H}^{\mathsf H}\mathbf{g}_i \mathbf{g}_i ^{\mathsf H} \mathbf{HQ}+ \mathbf{P}$, we have
\begin{align} \label{R_nsic_2}
\mathcal{R}_\text{nsic}=\sum_{k=1}^{K}\mathbbmss{E} \left\{\log\det(\mathbf{A})- \log\det(\mathbf{B}_k)\right\}.
\end{align}
By exploiting Jensen’s inequality, we obtain the following relationships:
\begin{align}
\mathbbmss{E} \left\{\log\det(\mathbf{A})\right\}\leq \log\det\left(\sum_{i=1}^{K}\mathbf{Q}^{\mathsf H} \mathbf{H}^{\mathsf H}\tilde{\mathbf{G}}_i \tilde{\mathbf{G}}_i^{\mathsf H} \mathbf{HQ}+ \mathbf{P}\right), \nonumber \\
\mathbbmss{E} \left\{\log\det(\mathbf{B}_k)\right\}\leq \log\det\left(\sum_{i\neq k}^{K}\mathbf{Q}^{\mathsf H} \mathbf{H}^{\mathsf H}\tilde{\mathbf{G}}_i \tilde{\mathbf{G}}_i^{\mathsf H} \mathbf{HQ}+ \mathbf{P}\right).  \nonumber
\end{align}
Instituting the above two inequations into \eqref{R_nsic_2} and after some basic mathematical manipulations, we can get the approximation of the sum rate with MMSE-nSIC decoding, i.e., $\tilde{\mathcal{R}}_\text{nsic}$ in \eqref{MMSE_nSIC_approximate}.

\subsection{Proof of Theorem \ref{Theorem_4}}\label{Proof_D}
To prove this theorem, we first introduce the following lemma.
\begin{lemma} \label{Lagrangian}
Problem $\mathcal{P}_\text{d}$ is equivalent to
\begin{subequations}
\begin{align}
{\mathcal P}_t: ~\max_{\mathbf{Q},\mathbf{W},\bm\rho} &~ \mathcal{F}_0 (\mathbf{Q},\mathbf{W},\bm\rho) =\sum_{k=1}^{K} \log_2(1+\rho_k) \nonumber\\
&~ -\sum_{k=1}^{K}\rho_k +\sum_{k=1}^{K}(1+\rho_k) \mathcal{B}_k \label{Pt_obj}\\
{\rm{s.t.}}&~ \eqref{cons_Q},\eqref{cons_Q2},~ \sum\nolimits_{k=1}^{K}\|\mathbf{HQ}\mathbf{w}_k\|^2\leq P_\text{max},
\end{align}
\end{subequations}
where $\mathcal{B}_k=\left(\sum_{i=1}^{K}\mathbf{w}_i^{\mathsf H}\mathbf{H}_k \mathbf{w}_i+N_k\right)^{-1}\mathbf{w}_k^{\mathsf H} \mathbf{H}_k\mathbf{w}_k$.
\end{lemma}
\begin{proof}
It is noted that $\mathcal{F}_0$ is concave over $\bm\rho$ with other variables fixed, thus the optimal $\bm\rho$ can be obtained by setting $\frac{\partial \mathcal{F}_0}{\partial \rho_k}=0$, which results in $\rho_k^{\star}=\frac{\mathcal{D}_k}{1-\mathcal{D}_k}$. By exploiting the Woodbury formula, we have $\rho_k^{\star}=\frac{\mathbf{w}_k^{\mathsf H}\mathbf{H}_k\mathbf{w}_k}{\sum_{i\neq k}^{K} \mathbf{w}_i^{\mathsf H}\mathbf{H}_k \mathbf{w}_i+N_k}$ with $\mathbf{H}_k=\mathbf{Q}^{\mathsf H}\mathbf{H}^{\mathsf H} \tilde{\mathbf{G}}_k \tilde{\mathbf{G}}_k^{\mathsf H} \mathbf{H} \mathbf{Q}$. It is worth noting that instituting $\rho_k^\star$ into \eqref{Pt_obj} recovers the objective function of $\mathcal{P}_\text{d}$ in \eqref{P3_obj}, which means that $\mathcal{P}_\text{d}$ and $\mathcal{P}_t$ have the same optimal objective value and optimal solution over $\mathbf{Q}$ and $\mathbf{W}$, i.e., $\max\limits_{\mathbf{Q},\mathbf{W},\bm\rho} \mathcal{F}_0 =\max\limits_{\mathbf{Q},\mathbf{W}} \tilde{\mathcal{R}}_\text{d}$ and  $\argmax\limits_{\mathbf{Q},\mathbf{W}} \mathcal{F}_0 (\mathbf{Q},\mathbf{W},\bm\rho_k^\star)=\argmax \limits_{\mathbf{Q},\mathbf{W}} \tilde{\mathcal{R}}_d$. These results suggest the two problems are equal.
\end{proof}
Lemma \ref{Lagrangian} establishes the equivalence of $\mathcal{P}_t$ and $\mathcal{P}_\text{d}$. Next we prove the equivalence of $\mathcal{P}_3$ and $\mathcal{P}_t$. Similarly, $\mathcal{F}_1$ is a concave differentiable function over $\bm\gamma_k$ with other variables fixed. Therefore, we let each $\frac{\partial \mathcal{F}_1}{\partial \bm\gamma_k}=0$ and the optimal $\bm\gamma_k^\star$ is given by $\bm\gamma_k^{\star}=\left(\sum_{i=1}^{K}\mathbf{w}_i^{\mathsf H}\mathbf{H}_k \mathbf{w}_i+N_k\right)^{-1} \tilde{\mathbf{G}}_k^{\mathsf H} \mathbf{H} \mathbf{Q} \mathbf{w}_k$. Inserting $\bm\gamma_k^\star$ back to $\mathcal{F}_1$ recovers the
objective function in \eqref{Pt_obj}, which means $\mathcal{P}_3$ and $\mathcal{P}_t$ are two equal problems, the equivalence between $\mathcal{P}_3$ and $\mathcal{P}_\text{d}$ is thus established.
\end{appendix}


\vspace{12pt}

\end{document}